\documentclass[preprint]{elsarticle}

\usepackage{amsmath,amsfonts,amssymb,amsthm}
\usepackage{mathrsfs,tikz}
\usepackage[latin1]{inputenc}
\usepackage[normalem]{ulem}

\newcommand{\gconf}{g.conf.~}
\newcommand{\HH}{\mathcal{H}}
\newcommand{\EE}{\mathcal{E}}
\newcommand{\CC}{\mathcal{C}}
\newcommand{\NP}{\mathcal{NP}}
\newcommand{\gd}{\gamma}
\newcommand{\grd}{\gamma_{\rm gr}}

\newcommand{\la}{\langle}
\newcommand{\ra}{\rangle}
\newcommand{\fp}{\mathfrak{fp}}
\newcommand{\ch}{\mathfrak{ch}}
\newcommand{\initial}{\bf 1,\bf 0}
\newcommand{\antes}{\vartriangleleft}
\newcommand{\noantes}{\not\vartriangleleft}
\def\tick{\tikz\fill[scale=0.4](0,.35) -- (.25,0) -- (1,.7) -- (.25,.15) -- cycle;}

\newtheorem{thm}{Theorem}[section]
\newtheorem{prop}[thm]{Proposition}
\newtheorem{obs}[thm]{Observation}
\newtheorem{lemma}[thm]{Lemma}
\newtheorem{cor}[thm]{Corollary}
\newtheorem{rem}[thm]{Remark}
\newtheorem{claim}{Claim}

\journal{}

\begin{document}

\begin{frontmatter}

\title{The polytope of legal sequences\tnoteref{grant}}

\author[a]{Manoel Camp\^elo}
\author[b,c]{Daniel Sever\'in\fnref{correspon}}

\address[a]{Dep. Estat\'istica e Matem\'atica Aplicada, Universidade Federal do Cear\'a, Brazil}

\address[b]{Depto. de Matem\'atica (FCEIA), Universidad Nacional de Rosario, Argentina}

\address[c]{CONICET, Argentina}

\tnotetext[grant]{Partially supported by grants PICT-2016-0410 (ANPCyT), PID ING538 (UNR),
                                                443747/2014-8, 305264/2016-8 (CNPq) and PNE 0112­00061.01.00/16 (FUNCAP/CNPq).\\
\emph{E-mail addresses}: \texttt{mcampelo@lia.ufc.br} (M. Camp\^elo),
                         \texttt{daniel@fceia.unr.edu.ar} (D. Sever\'in).}
\fntext[correspon]{Corresponding author at Departamento de Matem\'atica (FCEIA), UNR, Pellegrini 250, Rosario, Argentina.}

\begin{abstract}
A sequence of vertices in a graph is called a \emph{(total) legal dominating sequence} if every vertex in the sequence (total) dominates
at least one vertex not dominated by those ones that precede it, and at the end all vertices of the graph are (totally) dominated.
The \emph{Grundy (total) domination number} of a graph is the size of the largest (total) legal dominating sequence.
In this work, we address the problems of determining these two parameters by introducing a generalized version of them. We explicitly calculate the corresponding (general) parameter for paths and web graphs.
We propose integer programming formulations for the new problem and we study the polytope associated to one of them.
We find families of valid inequalities and derive conditions under which they are facet-defining.
Finally, we perform computational experiments to compare the formulations as well as to test valid inequalities as cuts in a B\&C framework.
\end{abstract}


\end{frontmatter}

\section{Introduction} \label{SSINTRO}

Covering problems are some of the most studied problems in Graph Theory and Combinatorial Optimization due to the
large number of applications.
Consider a hypergraph $\HH = (X, \EE)$ without isolated vertices.
An \emph{edge cover} of $\HH$ is a set of hyperedges $\CC \subset \EE$ that cover all vertices of $\HH$, i.e.$\!$ $\cup_{C \in \CC} C = X$.
The general covering problem consists in finding the \emph{covering number} of $\HH$, denoted by $\rho(\HH)$, which is the minimum number of hyperedges in an edge cover of $\HH$ \cite{BERGE}.

The most natural constructive heuristic for obtaining an edge cover of $\HH$ is as follows.
Start from empty sets $\CC$ and $W$ (the latter one keeps the already covered vertices).
At each step $i$, pick a hyperedge $C_i \in \EE$, add $C_i$ to $\CC$ and add all the elements of $C_i$ to $W$.
The process is repeated until $W = V$.
In addition, $C_i$ can only be chosen if at least one of its elements has not been previously included in $W$,
i.e.$\!$ $C_i \setminus (\cup_{j=1}^{i-1} C_j) \neq \emptyset$.

How bad can a solution given by this heuristic be (compared to the optimal one) ?
The answer leads to the concept of Grundy covering number of $\HH$, denoted by $\rho_{gr}(\HH)$, which computes the largest number of steps performed by such a constructive heuristic, or equivalently the largest number of hyperedges used in the resulting covering~\cite{BRESAR2014}. In general, ``Grundy numbers'' relate the worst case of greedy heuristics. They were initially addressed in the context of coloring problems~\cite{GRUNDYCOL}.

A particular case of the covering problem, motivated by a domination game \cite{P1,P2,P3}, occurs when $\HH$ is the hypergraph of the closed neighborhoods of vertices in a simple graph $G$, i.e.$\!$ $\HH = (V(G), \EE)$ where $\EE = \{N[v] : v \in V(G)\}$ and $N[v]\doteq \{u\in V(G): (u,v)\in E(G)\}\cup\{v\}$.
Here, the Grundy covering number of $\HH$ is called \emph{Grundy domination number} of $G$ \cite{BRESAR2014}.
Analogously, the \emph{Grundy total domination number} of $G$ is the Grundy covering number of the hypergraph of the open neighborhoods of vertices in $G$ (where the hyperedge related to $v\in V(G)$ is $N(v)\doteq N[v]\setminus \{v\}$)~\cite{BRESAR2016}.

The problems associated to these parameters are $\NP$-hard even when $G$ belongs to a restricted family of graphs.
For example, the Grundy domination problem is hard when $G$ is a chordal graph, although it is polynomial when $G$ is a tree,
a cograph or a split graph \cite{BRESAR2014}.
The total version of this problem is hard when $G$ is bipartite \cite{BRESAR2016} but it is polynomial on trees, $P_4$-tidy and
distance-hereditary bipartites \cite{NASINI2017}.

In this work, we introduce the General Grundy Domination Problem (GGDP), where the hyperedges can be either closed or open neighborhoods. We obtain the parameter associated with GGDP for paths and web graphs. We propose integer programming formulations for GGDP, which in particular
can obtain the Grundy (total) domination number of a graph.
We study the polytope associated to one of these formulations. We find families of valid inequalities and derive conditions under which they are facet-defining. Besides, we perform computational
experiments to compare the formulations as well as to test some of the proposed valid inequalities as cuts in a B\&C framework.
Some results contained in this work appeared without proof in the extended abstract \cite{LAGOS2017}.

\subsection{Definitions and notation}

Let $G = (V, E)$ be a simple graph on $n$ vertices and $C$ be a subset of vertices of $V$.
Define the function $N\la\_\ra : V \rightarrow \mathcal{P}(V)$ as follows: $N\la v \ra \doteq N(v)$ if $v \notin C$ and
$N\la v \ra \doteq N[v]$ if $v \in C$.
Assume that no vertex from $V \setminus C$ is isolated in $G$, so that $N\la v\ra\neq \emptyset$ for all $v\in V$. $N\la v\ra$ will be simply called the neighborhood of $v$.

Now, define $\HH(G;C)$ as the hypergraph $(V, \EE)$ with $\EE = \{N\la v\ra : v \in V\}$.
The Grundy covering number of $\HH(G;C)$ will be called \emph{Grundy domination number} of $G;C$, to be denoted $\grd(G;C)$.
The combinatorial problem of our interest is presented below:

\medskip

\noindent \textbf{GENERAL GRUNDY DOMINATION PROBLEM} (GGDP)\\
\noindent \underline{INSTANCE:} a graph $G = (V,E)$ and a set $C \subset V$ such that no isolated vertex is in $V \setminus C$.\\
\noindent \underline{OBJECTIVE:} obtain $\grd(G;C)$.

\medskip

Since the Grundy domination problems mentioned in the introduction are particular cases of our problem (note that $\grd(G;V)$ is the Grundy
domination number of $G$ while $\grd(G;\emptyset)$ is the Grundy total domination number of $G$), they can be simultaneously addressed by GGDP.

We refer to the pair ``$G;C$'' as an \emph{instance} and we present the same definitions given in \cite{BRESAR2014,BRESAR2016} in terms of instances (of the GGDP).

Let $S=(v_1,\ldots,v_k)$ be a sequence of distinct vertices of $G$.
The sequence $S$ is called a \emph{legal sequence} of an instance $G;C$ if
\begin{equation*}
W_i \doteq N\la v_i\ra \setminus \bigcup_{j=1}^{i-1} N\la v_j\ra \ne \emptyset
\end{equation*}
holds for every $i = 2,\ldots,k$.
If $S$ is a legal sequence, then we say that $v_i$ \emph{footprints} the vertices from $W_i$, and that $v_i$ is the \emph{footprinter} of every vertex $u\in W_i$.
That is, $v_i$ footprints a vertex $u \in N\la v_i\ra$ if $u$ does not belong to the neighborhood of $v_j$, for $j=1,2,\ldots,i-1$.

If $\cup_{j=1}^k W_j = V$ then $S$ is called a \emph{dominating sequence} of $G;C$. The legal dominating sequences of $G;C$ are in one-to-one correspondence with the solutions of the constructive heuristic for the covering number of $\HH(G;C)$.

Let $S$ be any legal sequence of $G;C$ with maximal length, i.e.$\!$ it is not possible to append a vertex at the end of $S$.
It is easy to see that $S$ is also a dominating sequence and, in particular, if its length is maximum then it is equal to $\grd(G;C)$.
Therefore, when determining $\grd(G;C)$, we do not need to restrict ourselves to legal dominating sequences.
Instead, we can only explore the structure of legal sequences, which is simpler.

For given vertices $v_1, v_2$, we denote $v_1 \antes v_2$ the expression
$N\la v_2\ra \setminus N\la v_1\ra \neq \emptyset$, which basically means that the sequence $(v_1, v_2)$ is legal.
Naturally, $v_1 \noantes v_2$ means $N\la v_2\ra \subset N\la v_1\ra$.
An instance $G;C$ is called \emph{clutter} if $\HH(G;C)$ is a clutter too, i.e.$\!$ if $u \antes v$ for all
distinct vertices $u, v$. It is called a \emph{strong clutter} if $|N\la u\ra \setminus N\la v\ra|\geq 2$, for all distinct vertices $u, v$.

We say that distinct vertices $u, v$ are \emph{twins} if $N\la u\ra = N\la v\ra$.
If $G;C$ does not have twin vertices, then $G;C$ is called \emph{twin free}. Clutters are twin free instances.

\section{Properties of the GGDP and examples} \label{SSPROP}

In order to reduce the size of the input graph, two simple rules can be applied.
In first place, if $G$ is the disjoint union of graphs $G_1$ and $G_2$, then:
$$\grd(G;C) = \grd(G_1;C \cap V(G_1)) + \grd(G_2; C \cap V(G_2)).$$
Therefore, we can restrict ourselves to connected graphs.

In addition, if there exist twin vertices $u, v$, then
$$\grd(G;C) = \grd(G - v; C \setminus \{v\}),$$
where $G - v$ is the graph obtained from $G$ by deleting $v$.
Hence, we can suppose that the instance is twin free.

Let $\delta(G;C)$ be the minimum cardinality of $N\la v\ra$ for all $v \in V(G)$.
In \cite{BRESAR2014,BRESAR2016}, upper bounds of Grundy domination numbers are presented.
They can easily be adapted to our problem:
\begin{prop} \cite{BRESAR2014,BRESAR2016}
$\grd(G;C) \leq m \doteq n-\delta(G;C)+1$.
\end{prop}
\begin{proof}
Let $S = (s_1,\ldots,s_k)$ be a legal dominating sequence of $G;C$ of maximum size and let $u$ be a vertex footprinted in the last step.
Since $u$ was not dominated in the previous steps, $N\la u\ra \cap \{s_1,\ldots,s_{k-1}\} = \emptyset$.
Then, $k - 1 \leq n - |N\la u\ra| \leq n - \delta(G;C)$, implying $k \leq m$.
\end{proof}

Now, we consider two families of graphs (paths and web graphs) that will serve as examples.
First, let $P_n$ denote an induced path on $n\geq 1$ vertices where $V(P_n)=\{1,\ldots,n\}$.
Note that, for any integer $n\geq 2$ and any $C \subset V(P_n)$,
$\delta(P_n;C) = 2$ when $\{1,n\} \subset C$ and $\delta(P_n;C) = 1$ otherwise. Therefore,
\begin{obs} \label{OBSERV1}
Let $n \geq 2$ and $C \subset V(P_n)$. If $\{1,n\} \subset C$ then $m = n - 1$; otherwise, $m = n$.
\end{obs}

We say that $C$ is a \emph{good configuration} (or ``\gconf\!'' for short) for $P_n$ if\\
{ \small
\indent \indent (i) $n=1$ and $C = \{1\}$,\\
\indent \indent (ii) $n=2$ and $C \neq \{1,2\}$,\\
\indent \indent (iii) $n\geq 3$ and either\\
\indent \indent \indent (iii.1) $1\notin C$ and $C$ is a \gconf for the subpath induced by $\{3,\ldots,n\}$ or\\
\indent \indent \indent (iii.2) $n\notin C$ and $C$ is a \gconf for the subpath induced by $\{1,\ldots,n-2\}$.\\ }

For the sake of simplicity, when we say that $C$ is a \gconf for a subpath $P'$ of $P$, we are actually referring to the set $C \cap V(P')$. Moreover, differently from the definition of GGDP, we allow the subgraph induced by $V(P')\setminus C$ to have isolated vertices in the \gconf definition.

From the previous observation, we trivially have
\begin{obs} \label{OBSERV2}
Let $n \geq 1$ and $C \subset V(P_n)$. If $C$ is a \gconf for $P_n$, then $m = n$.
\end{obs}

The following results give the Grundy domination number of a path and show where the upper bound $m$ is tight.

\begin{prop}
Let $n\geq 1$, $G = P_n$ and $C \subset V(G)$.
If $\{1,n\}\subset C$ or $C$ is a good configuration for $G$ then $\grd(G;C) = m$; otherwise, $\grd(G;C) = m - 1$.
\end{prop}
\begin{proof}
If $n=1$, we must have $C=\{1\}$, because $N\la1\ra\neq \emptyset$. So $C$ is a \gconf and $\grd(G;C) = m = 1$.
Now, consider the case $n = 2$.
If $C = \{1,2\}$ then $\grd(G;C) = m = 1$.
And, if $C \neq \{1,2\}$ ($C$ is a \gconf for $G$) then $\grd(G;C) = m = 2$.

For $n\geq 3$, note that $(1,2,\ldots,n-1)$ is a legal dominating sequence.
Indeed, $1$ footprints $2$ (and itself, if $1\in C$), $2$ footprints $3$ (and $1$, if $1\notin C$), and $i=3,\ldots,n-1$ footprints $i+1$.
Thus, $n-1\leq \grd(G;C)\leq m$.
If $\{1,n\}\subset C$, by Observation \ref{OBSERV1} we have $\grd(G;C) = m = n-1$.
It remains to consider the case $\{1,n\}\not\subset C$. By Observation \ref{OBSERV1}, $m = n$. Then, $\grd(G;C)\in \{m-1,m\}$,
and it is enough to prove that $C$ is a \gconf for $P_n$ if and only if there exists a legal dominating sequence of size $n$.

First, assume that $C$ is a \gconf for $P_n$.
We use induction on $n$.
Recall that we have already obtained $\grd(G;C)=2$ for $n=2$, and $\grd(G;C)=1$ is trivial for $n=1$.
For $n\geq 3$, (iii.1) or (iii.2) holds. Assume w.l.o.g.$\!$ that (iii.1) holds (the other case is symmetric).
The induction hypothesis ensures the existence of a legal dominating sequence $S$ for $(3,\ldots,n)$, with $|S|=n-2$.
Consider the extended sequence $S' = (1,S,2)$. $S'$ is legal and dominating for $P_n$ with $|S|=n$.

Now, assume that there exists a legal dominating sequence $S$ such that $|S|=n$.
We use again induction on $n$.
$C$ is trivially a \gconf for $P_n$ when $n\in\{1,2\}$.
For $n\geq 3$, any dominating legal sequence of length $n$ must start with an endpoint of $P_n$, say $1$, and such a vertex must belong to
$V\setminus C$. Moreover, $2$ must be the last vertex in the sequence, and the $n-2$ vertices between $1$ and $2$ define a legal dominating
sequence for $(3,\ldots,n)$. By the induction hypothesis, $C$ is a \gconf for $(3,\ldots,n)$. Since $1\notin C$, $C$ is also a
\gconf for $(1,\ldots,n)$.
\end{proof}	

\begin{cor} \label{cor:path}
Let $n\geq 1$, $G = P_n$ and $C \subset V(G)$.
If $C$ is a good configuration for $G$, then $\grd(G;C)=n$; otherwise, $\grd(G;C)=n-1$.
\end{cor}

The second nontrivial example is related to web graphs \cite{ANNEGRET}.
Let $n, k$ be positive integers such that $n \geq 2(k+1)$.
A \emph{web} $W_n^k = (V,E)$ is a graph with $V = \{0,\ldots,n-1\}$ and $E = \{(i,j) : 0<|i-j|\leq k \text{~or~} |i-j|\geq n-k\}$.
Note that $N(i) = \{i \ominus k, i \ominus (k-1), \ldots, i \ominus 1, i \oplus 1,\ldots, i \oplus (k-1), i \oplus k\}$, where $\oplus$ and $\ominus$ stand for the addition and subtraction modulo $n$.
Therefore, $m=n-2k$ if $C=V$, and $m=n-2k+1$ if $C\neq V$.

\begin{prop}
Let $G$ be the web graph $W_n^k$ and $C \subset V(G)$ ($C$ possibly empty).
$G;C$ is a clutter if and only if $n > 2(k+1)$ or $C=V$.
Moreover, $\grd(G;C) = m$ in the following cases: (i) $C=V$, or (ii) there is $i\in V\setminus C$ such that $V\setminus N[i]$ induces a path $P_t$, $t=n-2k-1$, and $C$ is a good configuration for $P_t$.
Otherwise $\grd(G;C) = m-1$.
\end{prop}
\begin{proof}
Suppose that $n = 2(k+1)$.
Then, $N(i) = V \setminus \{i, i \oplus (k+1)\} = N(i \oplus (k+1))$ for all $i$.
If $C \neq V$, let $i' \in V \setminus C$. Then, $N\la i'\ra \subset N\la i' \oplus (k+1)\ra$ and $G;C$ is not a clutter.
If $C = V$, $N\la i\ra \setminus N\la i \oplus (k+1)\ra = \{i\}$ for all $i \in V$.
In addition,  $j \oplus (k+1) \in N\la i\ra \setminus N\la j\ra$ for all $j \in V\setminus \{i, i \oplus  (k+1)\}$.
Therefore, $G;V$ is a clutter.

Now, suppose that $n > 2(k+1)$. Note that, for all $i,j$ such that $i \neq j$, there is
$k_{ij} \in N(i) \setminus N(j)$ such that $k_{ij} \notin \{i,j\}$. Therefore, $k_{ij} \in N\la i\ra \setminus N\la j\ra$
and $G;C$ is a clutter.

In order to determine $\grd(G;C)$, consider the sequence $(0, \ldots, n - 2k-1)$ of length $n-2k$.
It is a legal dominating sequence since vertex $0$ footprints $1,\ldots,k$ and $n-k,\ldots,n-1$ (and itself if $0 \in C$),
vertex $1$ footprints $k+1$ (and vertex $0$ if $0 \notin C$) and, if $n>2(k+1)$, vertex $i$ footprints $k+i$ for all $i = 2,\ldots,n-2k-1$. The last footprinted vertex is $n-k-1$. Therefore, we obtain $m-1\leq n-2k\leq \grd(G;C)\leq m$.
If $C=V$, we are done since $m=n-2k$.

Assume now that $C\neq V$. Let $i\in V\setminus C$, and $V_i\doteq V\setminus N[i] =\{i\oplus (k+1), i\oplus (k+2),\ldots,i\oplus (n-k-1)\}$. Note that $|V_i|=n-2k-1$. It suffices to show that $\grd(G;C) = m$ if and only if $V_i$ induces a path and $C$ is a good configuration for it. Recall that $m=n-2k+1$. The unique way of getting a dominating legal sequence of length $m$ is by starting with $i$ and then choosing a vertex that will footprint only one more vertex at each step. This means that, at any step but the last one, we cannot choose a vertex from $N(i)$ because it would footprint $i$ and at least one vertex from $V_i$. So, after $i$, we must choose all the $n-2k-1$ vertices in $V_i$, and finally a vertex from $N(i)$. In addition, $V_i$ must induce a path, otherwise some of its vertices would footprint at least 2 vertices.
Consider the sequence $S' = (i,S,j)$ where $j \in N(i)$ and $S$ is a maximum legal sequence of $G[V_i]$.
In virtue of Corollary~\ref{cor:path}, $|S|=n-2k-1$ if $C$ is a \gconf for $G[V_i]$ and $|S|=n-2k-2$ otherwise.
Therefore, $\grd(G;C) = m$ if and only if $V_i$ induces a path and $C$ is a \gconf for it.
\end{proof}

For example, consider the web graphs depicted in Figure \ref{fig:webs}.
Filled Circles denote the vertices of $C$,
i.e.$\!$ $C = \{1,2,3,4,5,7\}$. Upper bounds for $\grd(G;C)$ are $m = 7$ (if $G=W_8^3$) and $m = 3$ (if $G=W_8^1$).
In (a), neither $C$ is a \gconf for $V\setminus N[0] = \{2,3,4,5,6\}$ nor for $V\setminus N[6] = \{0,1,2,3,4\}$, so a legal sequence of maximum length is $(0,1,2,3,4,5)$.
In (b), since $C$ is a \gconf for $V\setminus N[0] = \{4\}$, a legal sequence of maximum length is $(0,4,1)$.
Note also that this latter instance is not a clutter since $N\la 0\ra \subset N\la 4\ra$. 
\begin{figure}
	\centering
		\includegraphics[scale=0.12]{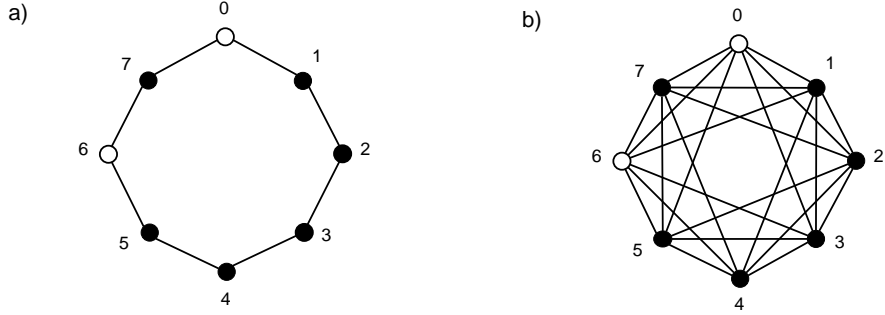}
	\caption{Web graphs: a) $W_8^1$, b) $W_8^3$.}  \label{fig:webs}
\end{figure}

\section{Integer programming formulations} \label{SSFORMU}

Legal sequences can be modeled as binary vectors as follows. For every $v \in V$ and $i = 1,\ldots,m$, let $y_{vi}$ be a binary value such
that $y_{vi} = 1$ if $v$ is chosen in step $i$. Also, for every $u \in V$ and $i = 1,\ldots,m$, let $x_{ui}$ be a binary value such that
$x_{ui} = 1$ if $u$ is available to be footprinted in step $i$ (i.e.$\!$ not footprinted by any of the chosen vertices in previous steps).
The following formulation computes the parameter $\grd(G;C)$:
\begin{align}
\max \sum_{i=1}^m \sum_{v \in V} y_{vi} & & \notag \\
\textrm{subject to} & & \notag \\
 & \sum_{v \in V} y_{vi} \leq 1, & \forall~~i = 1,\ldots,m \label{RESTR1}\\
 & \sum_{i=1}^m y_{vi} \leq 1, & \forall~~v \in V          \label{RESTR2}\\
 & y_{vi+1} \leq \sum_{u \in N\la v\ra} (x_{ui} - x_{ui+1}), & \forall~~v \in V,~ i = 1,\ldots,m-1 \label{RESTR3}\\
 & x_{ui} + \sum_{v \in N\la u\ra} y_{vi} \leq 1, & \forall~~u \in V,~ i = 1,\ldots,m \label{RESTR4}\\
 & x_{ui+1} \leq x_{ui}, & \forall~~u \in V,~i = 1,\ldots,m-1 \label{RESTR5}\\
 & x, y \in \{0,1\}^{n m}. & \notag
\end{align}
Constraints \eqref{RESTR1} ensure that at most one vertex is chosen in each step.
Constraints \eqref{RESTR2} guarantee that each vertex can be chosen no more than once.
Constraints \eqref{RESTR3} specify that $v$ can be chosen in the next step only if there is at least one non-footprinted vertex
of $N\la v\ra$ in the current one (that will become footprinted).
Constraints \eqref{RESTR4} force to footprint $u$ if some vertex of $N\la u\ra$ is chosen.
Finally, constraints \eqref{RESTR5} say that footprinted vertices will remain footprinted for the next steps.

Although this formulation works fine, there exist several integer solutions for each legal sequence.
Below, we present families of constraints that remove unnecessary solutions.
\begin{itemize}
\item A variable $x_{ui}$ can be set to zero even if it is not footprinted by any vertex. In order to forbid this situation, we consider these constraints:
\begin{align}
 & \sum_{v \in N\la u\ra} y_{v1} \geq 1 - x_{u1}, & \forall~~u \in V \label{RESTR6}\\
 & \sum_{v \in N\la u\ra} y_{vi+1} \geq x_{ui} - x_{ui+1}, & \forall~~u \in V,~i = 1,\ldots,m-1 \label{RESTR7}
\end{align}
\item It is allowed to choose no vertex in a step, leading to several symmetric solutions. They can be avoided by
replacing (\ref{RESTR1}) by
\begin{align}
 & \sum_{v \in V} y_{vi} = 1, & \forall~~i = 1,\ldots,{LB} \label{RESTR8}\\
 & \sum_{v \in V} y_{vi+1} \leq \sum_{v \in V} y_{vi}, & \forall~~i = {LB},\ldots,m-1 \label{RESTR9}
\end{align}
where $LB$ is a lower bound of $\grd(G;C)$.
These constraints force to assign vertices to the first $k$ steps when the solution represents a legal sequence of length $k$.
In addition, sequences of length smaller than $LB$ are removed.
\item We can impose that every integer solution represents a dominating sequence:
\begin{align}
 & \sum_{i=1}^m \sum_{v \in N\la u\ra} y_{vi} \geq 1, & \forall~~u \in V \label{RESTR10}
\end{align}
\end{itemize}

Here is a table of possible formulations. All of them include constraints \eqref{RESTR2}-\eqref{RESTR5}.
\begin{center}
\begin{tabular}{|c|c|c|c|c|c|}
\hline
 Form. & \eqref{RESTR1} & \eqref{RESTR6}-\eqref{RESTR7} & \eqref{RESTR8}-\eqref{RESTR9} & \eqref{RESTR10} & Solutions \\
\hline
 $F_1$ & \tick          &                               &                               &                 & 16253 \\
 $F_2$ & \tick          & \tick                         &                               &                 & 205 \\
 $F_3$ &                &                               & \tick                         &                 & 463 \\
 $F_4$ &                & \tick                         & \tick                         &                 & 43 \\
 $F_5$ & \tick          &                               &                               & \tick           & 1668 \\
 $F_6$ & \tick          & \tick                         &                               & \tick           & 124 \\
 $F_7$ &                &                               & \tick                         & \tick           & 68 \\
 $F_8$ &                & \tick                         & \tick                         & \tick           & 28 \\
\hline
\end{tabular}
\end{center}
There is a one-to-one correspondence between legal sequences of $G;C$ and solutions of $F_4$, and a one-to-one correspondence between dominating legal sequences and solutions of $F_8$.

The last column shows the number of integer solutions for the bull graph, i.e.$\!$ $V = \{1,\ldots,5\}$ and
$E = \{(1,2),(1,3),(2,3),(2,4),(3,5)\}$, with $C = V$, $m = 4$ and $LB = 1$.

Although one would expect that the fewer the number of solutions is, the smaller the size of the B\&B tree will be, it is known
that the addition of symmetry-breaking inequalities not always help to improve the optimization \cite{MARGOT}.
Therefore, we carried out computational experiments in order to determine which formulation performs better.
They are reported in Section \ref{SSCOMPU}.

\section{The polytope of legal sequences} \label{SSPOLY}

From now on, we suppose that $G$ has at least 3 vertices. We recall that $G$ is connected and $G;C$ is twin free.

The following easy result will be useful throughout the section:
\begin{lemma} \label{LEMITANV2}
Let $u \in V$.\\
(i)~ There exists $v \in N\la u\ra$ such that $|N\la v\ra| \geq 2$.\\
(ii)~ If $G;C$ is a clutter, $|N\la u\ra| \geq 2$.
\end{lemma}
\begin{proof}
In first place, note that $N\la v\ra \neq \emptyset$ and $N\la v\ra\neq \{v\}$, for every vertex $v$, since $G$ is a connected graph with at least 3 vertices.\\
(i)~Suppose that $N\la v\ra = \{u\}$ for all $v \in N\la u\ra$. Since $G$ does not have twin vertices, $|N\la u\ra| = 1$.
Therefore, $N\la u\ra = \{v\}$ for some $v \neq u$, and so $G$ is a $K_2$, which is absurd.\\
(ii)~Suppose that $u \in V$ such that $N\la u\ra = \{v\}$, and so $(u,v)$ is an edge.
Since $G;C$ is a clutter, $v \notin N\la w\ra$ for all $w\in V\setminus\{u\}$.
Then, $(u,v)$ is a connected component of $G$, which is also absurd.
\end{proof}

Let $P_i$ be the convex hull of the set of binary feasible solutions in formulation $F_i$.
In this section, we study the facial structure of $P_1$. We choose it for two reasons.
On the one hand, every valid inequality of this polytope is still valid for $P_i$ with $i \geq 2$.
On the other hand, the dimension of $P_1$ just depends on the size of the instance, which is an interesting feature for polyhedral studies.
The same does not happen with the other polytopes.

For instance, if we consider $P_2$ (recall that in this polytope $x_{ui-1}=1$ and $x_{ui}=0$ hold only if someone is footprinting $u$
in step $i$), then the equality $x_{u1} + \sum_{v \in N\la u\ra} y_{v1} = 1$ is valid for all $u \in V$, due to
\eqref{RESTR4} and \eqref{RESTR6}.
In addition, the equality $x_{vi+1} + y_{wi+1} = x_{vi}$ is valid for all $i = 1,\ldots,m-1$ and all pairs $v, w \in V$ such that
$N\la v\ra = \{w\}$.
Indeed, in this case, $x_{vi}\leq x_{vi+1}+y_{wi+1}\leq 1$ by \eqref{RESTR7} and \eqref{RESTR4}, respectively. If $x_{vi}=1$, the equality is clear. If $x_{vi}=0$, then $x_{vi+1}=0$ by \eqref{RESTR5} and $y_{wi+1}=0$ because $w$ is the only possible footprinter of $v$ and must be
chosen at step $i$ or before.

To analyse $P_3$ (in this polytope, sequences of size $k$ always use the first $k$ steps), we introduce a new parameter.
Let $i(G;C,v)$ be the index of the largest step where $v$ can occur in some legal sequence of $G;C$. Clearly, for all $v \in V$,
$i(G;C,v)$ is at most $\grd(G;C)$ and, in particular, $i(G;C,v)$ coincides with it for some $v$. Thus, obtaining these parameters is as hard as
determining the Grundy domination number itself.
In $P_3$, the equality $y_{vi} = 0$ is valid for all $v \in V$ and all $i > i(G;C,v)$.
In addition, if $k\in \{1,\ldots,LB\}$ and $V_k = \{v \in V: |N\la v\ra|  = n-k+1\}$, 
then $x_{ui} = 0$ for all $u\in V_k$ and all $i \geq k$. Indeed, no matter are the $k\leq LB$ first chosen vertices, some of them  will footprint any vertex in $V_k$.
In Prop.~\ref{PROPDIM3}, we give the dimension of $P_3$ for the case $LB = 1$.

Regarding $P_5$ (in this polytope, legal sequences are also dominating), the inequality $\sum_{i=1}^m \sum_{v \in V} y_{vi} \geq \gd(G)$ is valid, where $\gd(G)$ is the classical domination number of $G$. If $\gd(G) = \grd(G;C)$ such an inequality becomes an equality.
In \cite{BRESAR2014}, the authors partially characterize the family of graphs where $\gd(G) = \grd(G;V)$. 
Full characterization of this family is still an open question for $C = V$, as well as for $C = \emptyset$ \cite{BRESAR2016}.

Polytopes associated with the remaining formulations are even harder to study.

\subsection{Dimension}

For the sake of readability, write $P = P_1$ and assume $V = \{1, \ldots, n\}$.
Also, define $(\initial)$ as the solution $(x,y)$ where $x_{vi} = 1$ and $y_{vi} = 0$ for all $v \in V$ and $i = 1,\ldots,m$.
Note that it always belongs to $P$.

Below, we define some operations in order to manipulate solutions. Let $(x,y)$ be an integer solution of $P$.
\begin{itemize}
\item For a given set $S \subset V$ and $i = 1,\ldots,m$, the operation \emph{footprint}, denoted by $(x',y') = \fp(S,i;x,y)$, means that
vertices from $S$ are footprinted in step $i$. That is:
\begin{itemize}
\item $x'_{v,j} = 0$ for all $v \in S$ and $j = i,\ldots,m$
\item $x'_{v,j} = x_{v,j}$ for all $v \in S$ and $j = 1,\ldots,i-1$
\item $x'_{u,j} = x_{u,j}$ for all $u \in V\setminus S$ and $j = 1,\ldots,m$
\item $y'_{u,j} = y_{u,j}$ for all $u \in V$ and $j = 1,\ldots,m$
\end{itemize}
In the case that we need to footprint just one vertex $v$, we write $\fp(v,i;x,y)$.
\item For a given vertex $v$ and $i = 1,\ldots,m$, the operation \emph{choose}, denoted by $(x',y') = \ch(v,i;x,y)$, means that $v$ is
chosen in step $i$, and simply set $y'_{v,i} = 1$ to the solution $\fp(N\la v\ra,i;x,y)$.
\end{itemize}
We remark that, each time an operation is applied, we must check if $(x',y')$ is an integer solution of $P$, since
it is generally not guaranteed. The same operations can be used to propose points of $P_3$.

\begin{prop} \label{PROPDIM}
The polytope $P$ is full dimensional.
\end{prop}
\begin{proof}
Suppose that the following equation is valid for all points of $P$:
$$\sum_{i=1}^m \sum_{u \in V} \mu^x_{ui} x_{ui} + \sum_{i=1}^m \sum_{v \in V} \mu^y_{vi} y_{vi} = \mu_0$$
We prove that all of its coefficients are zero by giving pairs of points $(x^1,y^1), (x^2,y^2) \in P$ and using the fact
that $(\mu^x x^1 + \mu^y y^1) - (\mu^x x^2 + \mu^y y^2) = 0$ since both terms are equal to $\mu_0$.
\begin{enumerate}
\item Let $u \in V$ and consider the points $(x^1, y^1) = (\initial)$ and $(x^2, y^2) = \fp(u, m; \initial)$.
      Since the difference between both points lies just in $x_{um}$, we get $\mu^x_{um} = 0$.
\item Let $u \in V$ and $i = 1,\ldots,m-1$ and consider $(x^1, y^1) = \fp(u, i+1; \initial)$ and $(x^2, y^2) = \fp(u, i; \initial)$.
      We get $\mu^x_{ui} = 0$.
\item At this point $\mu^y y^1 - \mu^y y^2 = 0$ for any $(x^1,y^1), (x^2,y^2) \in P$ since $\mu^x = \bf 0$.
      Let $v \in V$ and $i = 1,\ldots,m$. Consider $(x^1, y^1) = (\initial)$ and $(x^2, y^2) = \ch(v, i; \initial)$.
      We get $\mu^y_{vi} = 0$.
\end{enumerate}
Therefore, $P$ is full dimensional.
\end{proof}

\begin{prop}  \label{PROPDIM3}
Let $LB = 1$ and $V_1 = \{v\in V: N \la v\ra = V\}$.
The dimension of $P_3$ is $m(n-|V_1|) + \sum_{v \in V} i(G;C,v)-1$ and a minimal system $\mathcal{M}$ of equalities is:
\begin{align*}
  & y_{vi} = 0 & \text{for all}~ v \in V ~\text{and}~ i = i(G;C,v) + 1, \ldots, m, \\
  & x_{ui} = 0 & \text{for all}~ u \in V_1 ~\text{and}~ i = 1, \ldots, m, \\
  & \sum_{v\in V} y_{v1}=1. &
\end{align*}
\end{prop}
\begin{proof}
We have seen that the equalities from $\mathcal{M}$ are all valid for $P_3$ and it is straightforward to prove that they are mutually
independent each other.
Thus, the stated expression gives an upper bound for $dim(P_3)$. Suppose that the following equality holds for all points of $P_3$:
$$\sum_{u \in V\setminus V_1} \sum_{i=1}^m  \mu^x_{ui} x_{ui} + \sum_{v \in V} \sum_{i=1}^{i(G;C,v)} \mu^y_{vi} y_{vi} = \mu_0$$
To complete the proof, it suffices to show that all the coefficients are zero, except for $\mu^y_{v1}$, $v\in V$, which have to be all equal to each other. We use an approach similar to the one applied in Prop.~\ref{PROPDIM}.

Note that, for a given $v \in V$, there exists a legal sequence $S_v$ such that $v$ is chosen in step $i(G;C,v)$.
Moreover, for a given $i = 1,\ldots,i(G;C,v)$, there exists a legal sequence $S_{vi}$ such that $v$ is chosen in step $i$ and $v$ is the last
element. $S_{vi}$ can be obtained from $S_v$ by dropping vertices.
\begin{enumerate}
\item Let $u \in V\setminus V_1$ and $i = 1,\ldots,m$. There exists $v\notin N\la u\ra$, and so $v$ does not footprint $u$. Consider $(x^1,y^1)=\ch(v,1;\initial)$ and $(x^2, y^2) = \fp(u, i; x^1,y^1)$. We obtain $\sum_{j=i}^m \mu^x_{uj} = 0$. By assigning $i$ in the order
$i=m,m-1,\ldots,1$, we get $\mu^x_{ui} = 0$ for all $i$.
\item At this point, $\mu^x = \bf 0$.
Let $w \in V$. We prove that $\mu^y_{v1}=\mu^y_{w1}$ for all $v \in V\setminus \{w\}$ by taking the points
$\ch(v,1;\initial)$ and $\ch(w,1;\initial)$.
\item Let $v \in V$ and $i = 2, \ldots, i(G;C,v)$. Consider $(x^1, y^1)$ as the point describing $S_{vi}$ and $(x^2,y^2)$ as the point representing $S_{vi}$ without $v$ (its last element). Note that $(x^2,y^2)$ still satisfies \eqref{RESTR8} because $i\geq 2$.
Using $\mu^x = \bf 0$, we get $\mu^y_{vi} = 0$.
\end{enumerate}
\end{proof}

\subsection{Facet-defining inequalities}

We start by mentioning some basic properties about the general facial structure of $P$.
\begin{prop} \label{prop:twins}
Let $\pi^x x + \pi^y y\leq \pi_0$ be a valid inequality for $P$. Then:\\
(i)~$\pi_0 \geq 0$;\\
(ii)~$\sum_{v \in S} \sum_{k=1}^{j(v)} \pi^x_{vk} \leq \pi_0$ for all $S \subset V$ and $j:S\rightarrow \{1,\ldots,m\}$.\\
Moreover, if $\pi^x x + \pi^y y\leq \pi_0$ is facet-defining for $P$, then:\\
(iii)~$\pi^y_{vj} \geq 0$, if the inequality is different from $-y_{vj}\leq 0$, for all $v \in V$ and $j = 1,\ldots,m$.
\end{prop}
\begin{proof}
(i)~ Since $({\bf 0}, {\bf 0}) \in P$, $0 = \pi^x {\bf 0} + \pi^y {\bf 0} \leq \pi_0$.\\
(ii)~~ Consider the point $(x,{\bf 0})$ such that:
$x_{sk} = 1$ for all $s \in S$ and $k = 1,\ldots,j(s)$;
$x_{sk} = 0$ for all $s \in S$ (such that $j(s) \neq m$) and $k = j(s)+1,\ldots,m$;
$x_{vk} = 0$ for all $v \in V \setminus S$ and $k = 1,\ldots,m$.
Clearly $(x,{\bf 0}) \in P$. Then, $\sum_{v \in S} \sum_{k=1}^{j(v)} \pi^x_{vk} = \pi^x x + \pi^y {\bf 0} \leq \pi_0$.\\
(iii)~Let $F$ be the facet defined by $\pi^x x + \pi^y y\leq \pi_0$.
Since the inequality is different from $-y_{vj}\leq 0$, there is a point $(x^1,y^1) \in F$ such that $y^1_{vj}=1$.
Let $y^2$ be the vector obtained from $y^1$ by setting $y^1_{vj}=0$.
Clearly $(x^1,y^2)\in P$. Then, $\pi^x x^1 + \pi^y y^2\leq \pi_0$.
Also, $\pi^x x^1 + \pi^y y^1= \pi_0$, which implies $\pi^y(y^2-y^1)= - \pi^y_{vj} \leq 0$.
\end{proof}

We use the following technique for proving sufficient conditions for a given inequality $\alpha^x x + \alpha^y y \leq \alpha_0$ to be
facet-defining of $P$.
Let $F$ be the face defined by the inequality. Hence, $\alpha^x x + \alpha^y y = \alpha_0$ for all $(x,y) \in F$.
If $\mu^x x + \mu^y y = \mu_0$ is another equation valid for all points of $F$, provide a real number $\beta$ and prove that this latter
equality can always be obtained from the product between $\beta$ and the former one.

Necessary conditions for $\alpha^x x + \alpha^y y \leq \alpha_0$ to be facet-defining are proved as follows: negate the condition, then
give a stronger inequality that dominates it or provide an equality (not equivalent to $\alpha^x x + \alpha^y y = \alpha_0$) that is
satisfied by all the points of $F$.

\begin{prop} \label{EASYFACETS}
The following inequalities define a facet of $P$:\\
(i)~ $y_{vi} \geq 0$ for all $v \in V$ and $i = 1,\ldots,m$\\
(ii)~ $x_{ui} \geq 0$ for all $u \in V$ if and only if $i = m$\\
(iii)~ Constraints \eqref{RESTR5} for each $u\in V$ and $i=1,\ldots,m-1$ if and only if for all $v \in V$, $N\la v\ra \neq \{u\}$
      (i.e.$\!$ every leaf of $G$ adjacent to $u$ must belong to $C$).
\end{prop}
\begin{proof}
\noindent (i)~ Here, $F = \{ (x,y) \in P : y_{vi} = 0\}$. Let $\beta = \mu^y_{vi}$.
Then, we have to prove that $\mu^x = \bf 0$, and $\mu^y_{v'i'} = 0$ for all $(v',i') \neq (v,i)$.
Items 1 and 2 given in the proof of dimension (Prop.~\ref{PROPDIM}) show again that $\mu^x = \bf 0$. At this point
$\mu^y y^1 - \mu^y y^2 = 0$ for any $(x^1,y^1), (x^2,y^2) \in P$.
Let $v' \in V$ and $i' = 1,\ldots,m$ such that $(v',i') \neq (v,i)$ and consider $(x^1, y^1) = (\initial)$ and
$(x^2, y^2) = \ch(v', i'; \initial)$. We get $\mu^y_{v'i'} = 0$.\\

\noindent (ii)~ Suppose that $i = m$. Here, $F = \{ (x,y) \in P : x_{um} = 0\}$. Let $\beta = \mu^x_{um}$.
Then, we have to prove these cases:
1) $\mu^x_{u'i'} = 0$ for all $(u',i') \neq (u,m)$, and
2) $\mu^y_{v'i'} = 0$ for all $(v',i')$.
Consider the point $(x^0 ,y^0) = \fp(u, m; \initial)$.
For each case, we write the condition and the points involved to prove it:
\begin{enumerate}
\item $\mu^x_{u'i'} = 0$ for all $(u',i') \neq (u,m)$.\\
			Let $u' \in V\setminus\{u\}$. Consider the points $(x^1, y^1) = (x^0, y^0)$ and $(x^2, y^2) = \fp(u', m; x^0, y^0)$.
      We get $\mu^x_{u'm} = 0$.
			Now, let $u' \in V$ and $i' = 1,\ldots,m-1$ and consider $(x^1, y^1) = \fp(u', i'+1; x^0, y^0)$ and $(x^2, y^2) = \fp(u', i'; x^0, y^0)$.
      We get $\mu^x_{u'i'} = 0$.
\item $\mu^y_{v'i'} = 0$ for all $(v',i')$.\\
			Let $v' \in V$ and $i' = 1,\ldots,m$. Consider $(x^1, y^1) = (x^0, y^0)$ and $(x^2, y^2) = \ch(v', i'; x^0, y^0)$.
      By using equalities obtained in item 1, we get $\mu^y_{v'i'} = 0$.
\end{enumerate}
To see that the inequality is not facet-defining when $i < m$, note that it is dominated
by constraints \eqref{RESTR5}.\\

\noindent (iii)~ Here, $F = \{ (x,y) \in P : x_{ui+1} - x_{ui} = 0\}$. Let $\beta = -\mu^x_{ui}$.
Then, we have to prove these cases:
1) $\mu^x_{ui+1} = -\mu^x_{ui}$,
2) $\mu^x_{u'i'} = 0$ for all $(u',i') \notin \{(u,i),(u,i+1)\}$, and
3) $\mu^y_{v'i'} = 0$ for all $(v',i')$.
\begin{enumerate}
\item $\mu^x_{ui+1} = -\mu^x_{ui}$.\\
			Consider the points $(x^1, y^1) = \fp(u, i; \initial)$ and $(x^2, y^2) = \fp(u, i+2; \initial)$ if $i+2 \leq m$ or
			$(x^2, y^2) = (\initial)$ otherwise. Then, $\mu^x x^1 + \mu^y y^1 - \mu^x x^2 - \mu^y y^2 = - \mu^x_{ui+1} - \mu^x_{ui} = 0$.
\item $\mu^x_{u'i'} = 0$ for all $(u',i') \notin \{(u,i),(u,i+1)\}$.\\
			Let $u' \in V$ and $i' = 1,\ldots,m$ such that $(u',i') \notin \{(u,i),(u,i+1)\}$. Consider the points
			$(x^1, y^1) = \fp(u', i'; \initial)$ and $(x^2, y^2) = \fp(u', i'+1; \initial)$ if $i < m$ or $(x^2, y^2) = (\initial)$ otherwise.
			We get $\mu^x_{u'i'} = 0$.
\item $\mu^y_{v'i'} = 0$ for all $(v',i')$.\\
			Let $v' \in V$ and $i' = 1,\ldots,m$. Consider $(x^1, y^1) = \fp(u, i; \initial)$ and $(x^2, y^2) = \ch(v', i'; x^1, y^1)$.
			By using equalities obtained in item 2, we get $\mu^y_{v'i'} = 0$.
			Note that, even in the case that $u \in N\la v'\ra$ and $i' > i$, we have $(x^2, y^2) \in F$ since $N\la v'\ra\neq \{u\}$ by
			hypothesis and, therefore, $v'$ has a vertex to footprint from $N\la v'\ra \setminus \{u\}$ in step $i'$.
\end{enumerate}
To see that the inequality does not define a facet when $u$ is adjacent to a leaf $v\in V\setminus C$, note that it is dominated
by $x_{ui+1}+y_{vi+1}\leq x_{ui}$, which is constraint \eqref{RESTR3} for $v$.
\end{proof}

\begin{cor} \label{EASYFACETSCLUTTER}
If $G;C$ is a clutter, constraints \eqref{RESTR5} always define facets of $P$.
\end{cor}
\begin{proof}
Straightforward from Lemma \ref{LEMITANV2}.(ii).
\end{proof}

Constraints~\eqref{RESTR1} can be generalized as follows.
Define $V^\subset\doteq \{u\in V: v \noantes u ~~\forall v\in V\}$ and $V^\supset\doteq \{w\in V: w \noantes v ~~\forall v\in V\}$.
Note that $v\in V^\supset$ is equivalent to $N\la v\ra=V$ (in other words, $V^\supset$ is the set of universal vertices of
$G$ belonging to $C$). Note also that $v,v'\in V^\subset$ or $v,v'\in V^\supset$ imply $N\la v\ra =N\la v'\ra$. As there are no twin vertices, we must have $|V^\subset|\leq 1$ and $|V^\supset|\leq 1$.

\begin{prop} \label{RESTR1STRONGER}
Let $i = 1,\ldots,m$.
If $V^\subset = \emptyset$ or $i = m$, and $V^\supset = \emptyset$, constraint~\eqref{RESTR1} for $i$ defines a facet of $P$.
Otherwise, the following stronger inequalities are facet-defining for $P$:
\begin{itemize}
\item If $V^\subset = \{u\}$, $i < m$ and $V^\supset = \emptyset$ then
\begin{equation} \label{RESTR1STRONGER1}
\sum_{v \in V} y_{vi} + \sum_{j=i+1}^m y_{uj} \leq 1.
\end{equation}
\item If $V^\subset = \emptyset$ or $i = m$, and $V^\supset = \{w\}$ then
\begin{equation} \label{RESTR1STRONGER2}
x_{wi} + \sum_{j=1}^{i-1} y_{wj} + \sum_{v \in V} y_{vi} \leq 1.
\end{equation}
\item If $V^\subset = \{u\}$, $i < m$ and $V^\supset = \{w\}$ then
\begin{equation} \label{RESTR1STRONGER3}
x_{wm} + \sum_{j=1}^{i-1}y_{wj} + \sum_{v \in V} y_{vi} + \sum_{j=i+1}^m y_{uj}\leq 1.
\end{equation}
\end{itemize}
where, in \eqref{RESTR1STRONGER2} and \eqref{RESTR1STRONGER3}, $\sum_{j=1}^{i-1}y_{wj}$ vanishes when $i = 1$.
\end{prop}
\begin{proof}
Let $(x,y)$ be an integer point of $P$.
The validity of \eqref{RESTR1STRONGER1} is a direct consequence of $v \noantes u$ for all $v \in V$.
For \eqref{RESTR1STRONGER2}, we have $\sum_{j=1}^{i-1} y_{wj} + \sum_{v \in V} y_{vi} \leq 1$ as a consequence of
$w \noantes v$ for all $v \in V$. If this inequality holds strictly, \eqref{RESTR1STRONGER2} is trivially satisfied because $x_{wi}\leq 1$.
If $\sum_{j=1}^{i-1} y_{wj} + \sum_{v \in V} y_{vi} = 1$, a vertex was chosen in some step $j=1,\ldots,i$. As $N\la w\ra = V$, $w$ is
footprinted by that vertex and $x_{wi}=0$. So, \eqref{RESTR1STRONGER2} is valid. A similar reasoning can be made for \eqref{RESTR1STRONGER3}.

Let $k = i$ if $V^\subset\neq\emptyset$ and $k = m$ otherwise, and $F = \{ (x,y) \in P : x_{wk} + \sum_{j=1}^{i-1} y_{wj} +
\sum_{v \in V} y_{vi} + \sum_{j=i+1}^m y_{uj} = 1\}$. Here, $x_{wk} + \sum_{j=1}^{i-1} y_{wj}$ vanishes if $V^\supset = \emptyset$,
$\sum_{j=1}^{i-1} y_{wj}$ vanishes if $i=1$, and $\sum_{j=i+1}^m y_{uj}$ vanishes if $V^\subset = \emptyset$ or $i = m$.
The points lying in $F$ are such that:
\emph{a}) some vertex is chosen at step $i$ or
\emph{b}) if $V^\subset=\{u\}$, $u$ is chosen in some step from $i+1,\ldots,m$; c) if $V^\supset=\{w\}$, either $w$ is not footprinted at step $k$ or $w$ is chosen in some step from $1,\ldots,i-1$. 

Let $v \in V\setminus\{u\}$ and $\beta = \mu^y_{vi}$. We have to prove these cases:
1) $\mu^x_{u'i'} = 0$ for all $u'\in V\setminus V^\supset$, $i'=1,\ldots,m$,
2-3) if $V^\supset=\{w\}$, then $\mu^x_{wi'} = 0$ for all $i'\neq k$, and $\mu^x_{wk}=\mu^y_{vi}$,
4-5) $\mu^y_{v'i'} = 0$ for all $v' \in V\setminus V^\supset$ and $i' < i$ and for all $v' \in V\setminus V^\subset$ and $i'> i$,
6) $\mu^y_{v'i} = \mu^y_{vi}$ for all $v' \in V \setminus \{v\}$,
7) if $V^\supset=\{w\}$, $\mu^y_{wj} = \mu^y_{vi}$ for all $j<i$,
8) if $V^\subset=\{u\}$, $\mu^y_{uj} = \mu^y_{vi}$ for all $j>i$.
\begin{enumerate}
\item $\mu^x_{u'i'} = 0$ for all $u'\in V\setminus V^\supset$, $i' = 1,\ldots,m$.\\
			As $u'\notin V^\supset$ (or equivalently $N\la u'\ra\neq V$), there exists $v' \notin N\la u'\ra$.
			Consider the points $(x^1, y^1) = \fp(u', i'+1; \ch(v',i; \initial))$ if $i' < m$ or $(x^1, y^1) = \ch(v',i; \initial)$ otherwise,
			and $(x^2, y^2) = \fp(u', i'; x^1, y^1)$. We get $\mu^x_{u'i'} = 0$.
\item $\mu^x_{wi'} = 0$ for all $i' = 1,\ldots,m$, $i'\neq k$ (when $V^\supset = \{w\}$).\\
			First, suppose $k=m$. Hence, $V^\subset=\{u\}$ and $i'<m$. Pick a $v \in V$ such that $v=w$ if $i'\leq i-2$ and $v=u$ if $i'\geq i$.
			Consider the points $(x^1, y^1) = \ch(v,i'+1; \initial))$ and $(x^2, y^2) = \fp(w, i'; x^1, y^1)$. We get $\mu^x_{wi'} = 0$.
			Now, assume $k=i$. If $i'<i$, the previous two points show $\mu^x_{wi'} = 0$. If $i'>i$, consider the points
			$(x^1, y^1) = \fp(w, i'; \initial)$, and $(x^2, y^2) = \fp(w, i'+1; \initial)$ if $i'<m$ or $(x^2, y^2) = (\initial)$ otherwise.
			Again, we get $\mu^x_{wi'} = 0$.
\item $\mu^x_{wk} = \mu^y_{vi}$ (when $V^\supset = \{w\}$).\\
			 The points $(x^1, y^1) = \ch(v, i; \initial)$ and $(x^2, y^2) = (\initial)$ belong to $F$. In particular, $x^1_{wk}=x^2_{vi}=0$ and $x^1_{vi}=x^2_{wk}=1$. By items 1 and 2, we get $\mu^y_{vi}=\mu^x_{wk}$.
\item $\mu^y_{v'i'} = 0$ for all $v' \in V\setminus V^\supset$ and $i' < i$. \\
			Since $v'\notin V$, there is $v''$  such that $v'\antes v''$. 
			As $i'<i$, the points $(x^1, y^1) = \ch(v'',i; \initial)$ and $(x^2, y^2) = \ch(v',i'; x^1, y^1)$ belong to $F$. Then, by item 1, $\mu^y_{v'i'} = 0$.
\item $\mu^y_{v'i'} = 0$ for all $v' \in V\setminus V^\subset$ and $i'> i$.\\
			Since $v'\notin V^\subset$, there is $v''$ such that $v''\antes v'$. 
			As $i'>i$, the points $(x^1, y^1) = \ch(v'',i; \initial)$ and $(x^2, y^2) = \ch(v',i'; x^1, y^1)$ are in $F$. In particular, if $V^\supset\neq\emptyset$, i.e. $V^\supset=\{w\}$, then $x^1_{wk}=x^2_{wk}=0$. Therefore, items~1-2 lead to $\mu^y_{v'i'} = 0$. 
\item $\mu^y_{v'i} = \mu^y_{vi}$ for all $v' \in V \setminus \{v\}$.\\
			The points $(x^1, y^1) = \ch(v',i; \initial)$, $(x^2, y^2) = \ch(v,i; \initial)$, both in $F$, and items~1-2. Note that $x^1_{wk}=x^2_{wk}=0$ 
			allow us to obtain $\mu^y_{v'i} = \mu^y_{vi}$. 
\item $\mu^y_{wj} = \mu^y_{vi}$ for all $j=1,\ldots,i-1$ (when $V^\supset = \{w\}$).\\
			Similar to item 6 with points $(x^1, y^1) = \ch(w,j; \initial)$ and $(x^2, y^2) = \ch(v,i; \initial)$.
\item $\mu^y_{uj} = \mu^y_{vi}$ for all $j=i+1,\ldots,m$ (when $V^\subset = \{u\}$).\\
			Similar to item 6 with points $(x^1, y^1) = \ch(u,j; \initial)$ and $(x^2, y^2) = \ch(v,i; \initial)$.
\end{enumerate}
\end{proof}

\begin{cor} \label{RESTR1STRONGERCLUTTER}
If $G;C$ is a clutter, constraints \eqref{RESTR1} define facets of $P$.
\end{cor}

Constraints (\ref{RESTR3}) are a particular case of the following ones:

\begin{prop} \label{GENERALINEQ1}
Let $t \geq 1$ and $W \doteq \{w_1, \ldots, w_t\} \subset V$ satisfying the property $w_1 \noantes w_j$ for all
$j = 2,\ldots,t$. Then, for all $i = 1,\ldots,m-1$, the following is a valid inequality:
\begin{equation} \label{NOVA0}
  \sum_{w \in W} y_{w i+1} \leq \sum_{u \in N\la w_1\ra} (x_{ui} - x_{ui+1}).
\end{equation}
Moreover, it is facet-defining if, and only if, every $v\in V \setminus W$ satisfies $w_1 \antes v$
($W$ is \emph{maximal} with respect to the property).
\end{prop}
\begin{proof}
First, we prove validity.
If the r.h.s. is one or more, the inequality is dominated by a constraint \eqref{RESTR1}.
Now, suppose that the r.h.s. is 0. Therefore $x_{ui} = x_{ui+1}$ for all $u \in N\la w \ra$ for every $w \in W$.
Thus, none of the vertices of $W$ can be a footprinter in the step $i + 1$, and the inequality is valid.

We now prove that it is facet-defining if and only if $W$ is maximal. If there is $v'\in V\setminus W$ such that
$w_1 \noantes v'$, we can add $y_{v'i+1}$ to the l.h.s. of \eqref{NOVA0} to strengthen the inequality.
In other words, the equality $y_{v'i+1} = 0$ is valid in the face defined by \eqref{NOVA0}.

Now, assume that $W$ is maximal.
Let $F = \{ (x,y) \in P : \sum_{w \in W} y_{w i+1} - \sum_{u \in N\la w_1\ra} x_{ui} + \sum_{u \in N\la w_1\ra} x_{ui+1} = 0\}$. 
First, we characterize the points in $F$:
\emph{a}) if no vertex from $W$ is chosen in step $i+1$ then every vertex from $N\la w_1\ra$ must not be footprinted in step $i+1$ or
\emph{b}) if a vertex $w_j \in W$ is chosen in step $i+1$ then exactly one vertex from $N\la w_j\ra$ must be footprinted in this step
while the remaining ones must have been footprinted in a previous step, also every vertex from $N\la w_1\ra \setminus N\la w_j\ra$ must not be
footprinted in step $i+1$. 

Let $\beta = \mu^y_{w_1 i+1}$. Below, we enumerate the conditions needed to prove facet-definition and the points involved to obtain them:
\begin{enumerate}
\item $\mu^y_{w_j i+1} = \mu^y_{w_1 i+1}$ for all $j = 2,\ldots,t$.\\
			Let $u_j$ be some vertex of $N\la w_j\ra$. Consider
			$(x^1, y^1) = \ch(w_j, i+1; \fp(N\la w_1\ra \setminus \{u_j\}, i; \initial))$ and
			$(x^2, y^2) = \ch(w_1, i+1; \fp(N\la w_1\ra \setminus \{u_j\}, i; \initial))$.
\item $\mu^x_{ui} = -\mu^y_{w_1 i+1}$ for all $u \in N\la w_1\ra$.\\
			We use $(x^1, y^1) = \ch(w_1, i+1; \fp(N\la w_1\ra \setminus \{u\}, i; \initial))$ and
			$(x^2, y^2) = \fp(N\la w_1\ra, i; \initial)$.
\item $\mu^x_{u i+1} = \mu^y_{w_1 i+1}$ for all $u \in N\la w_1\ra$.\\
			Consider $(x^1, y^1) = \fp(u, i; \initial)$, $(x^2, y^2) = \fp(u, i+2; \initial)$ if $i+2 \leq m$ or $(x^2, y^2) = (\initial)$ otherwise.
			We get $\mu^x_{ui} + \mu^x_{ui+1} = 0$. Then, apply equality of item 2.
\item $\mu^x_{u'i'} = 0$ for all $(u',i') \notin \{(u,i), (u,i+1) : u \in N\la w_1\ra\}$.\\
			Consider the points $(x^1, y^1) = \fp(u',i';\initial)$, and $(x^2, y^2) = \fp(u',i'+1;\initial)$ if $i < m$ or $(x^2, y^2) = (\initial)$ otherwise.
\item $\mu^y_{v'i'} = 0$ for all $(v',i') \notin \{(w,i+1) : w \in W\}$.\\
			If $v'\notin W$, the hyphotesis ensures that $w_1 \antes v'$.
			Then, besides $(x^1, y^1) = \fp(N\la w_1\ra, i; \initial)$, the point $(x^2, y^2) = \ch(v',i';x^1,y^1)$ also belongs to $F$
			since there exists a vertex from $N\la v'\ra \setminus N\la w_1 \ra$ that can be footprinted in step $i'$.
			By item 4, we get $\mu^y_{v'i'} = 0$.
			If $v'\in W$, then $i'\neq i+1$. In this case, take the points $(x^1, y^1) = (\initial)$ and $(x^2, y^2) = \ch(v',i'; \initial)$.
			We get $\mu^y_{v'i'} - \sum_{u\in N\la v'\ra} \sum_{j=i'}^m \mu^x_{uj} = 0$. By item 4, we get $\mu^y_{v'i'} = 0$ if $i'>i+1$.
			Otherwise, we obtain $\mu^y_{v'i'} - \sum_{u\in N\la v'\ra} \mu^x_{ui} - \sum_{u\in N\la v'\ra} \mu^x_{ui+1} = 0$.
			Since $w_1 \noantes v'$, we can apply items 2 and 3 to get $\mu^y_{w_1 i+1} + \mu^y_{w_1 i+1}=0$, and so $\mu^y_{v'i'}=0$.
\end{enumerate}
\end{proof}
\begin{cor} \label{GENERALINEQ1CLUTTER}
If $G;C$ is a clutter, constraints \eqref{RESTR3} always define facets of $P$.
\end{cor}

The following proposition tells when constraints \eqref{RESTR4} define facets.

\begin{prop} \label{Prop-RESTR4}
Constraint \eqref{RESTR4} defines a facet of $P$ if and only if $i = 1$ or for every $v \in N\la u\ra$ there exists
$w \in N\la u\ra \setminus \{v\}$ such that $v \antes w$.
\end{prop}
\begin{proof}
Consider the face $F = \{ (x,y) \in P : x_{ui} + \sum_{v \in N\la u\ra} y_{vi} = 1\}$.
First, we characterize those points in $F$:
\emph{a}) $u$ is not footprinted in steps $1,\ldots,i$ or
\emph{b}) some vertex from $N\la u\ra$ is chosen in step $i$.

Let $\beta = \mu^x_{ui}$. Below, we enumerate the conditions needed to prove facet-definition and the points involved to obtain them:
\begin{enumerate}
\item $\mu^x_{u'i'} = 0$ for all $(u',i') \neq (u,i)$.\\
			Let $u'\in V\setminus \{u\}$ and $i'=1,\ldots,m$. The points $(x^1,y^1) = \fp(u',i'; \initial)$ and
			$(x^2,y^2) = \fp(u',i'+1;\initial)$ if $i' < m$ or $(x^2, y^2) = (\initial)$ otherwise, show that $\mu^x_{u'i'} = 0$.
			The same points can be used to show that $\mu^x_{ui'} = 0$ when $i' > i$.
			Now, let $i' < i$. Take a vertex $v\in N\la u \ra$ such that $|N\la v \ra| \geq 2$.
			Such a $v$ always exists by Lemma \ref{LEMITANV2}.(i).
			The points $(x^1,y^1) = \ch(v,i;\initial)$ and $(x^2,y^2) = \fp(u,i'; x^1,y^1)$ show that $\sum_{j=i'}^{i-1} \mu^x_{uj} = 0$.
			By assigning $i'$ in the order $i-1, i-2, \ldots$, we obtain $\mu^x_{ui-1} = 0$, $\mu^x_{ui-2} = 0$, $\ldots$.
			Note that $(x^2, y^2) \in F$ since $v$ has a vertex to footprint from $N\la v\ra \setminus \{u\}$ in step $i$.
\item $\mu^y_{vi} = \mu^x_{ui}$ for all $v \in N\la u\ra$.\\
			To get $\mu^y_{vi} = \mu^x_{ui}$ we can use points $(x^1,y^1) = \ch(v,i;\initial)$,
			$(x^2,y^2) = (\initial)$ and apply item 1.
\item $\mu^y_{v'i'} = 0$ for all $(v',i')$ such that $v' \notin N\la u\ra$ or $i' \neq i$.\\
			Let $v'\in V\setminus N\la u\ra$. Hence, $u$ can not be footprinted by $v'$.
			We obtain $\mu^y_{v'i'} = 0$ with points $(x^1,y^1) = \ch(v',i';\initial)$ and $(x^2,y^2) = (\initial)$ and by applying item 1.
			Now, let $v' \in N\la u\ra$. If $i' > i$, the same points can be used to show that $\mu^y_{v'i'} = 0$.
			If $i' < i$, by hypothesis there exists $w \in N\la u\ra\setminus\{v'\}$ such that $v' \antes w$.
			Consider $(x^1,y^1) = \ch(w,i;\initial)$ and $(x^2,y^2) = \ch(v',i';x^1, y^1)$. We obtain $\mu^y_{v'i'} = 0$.
			Note that $(x^2, y^2) \in F$ since $w$ has a vertex to footprint in step $i$ not footprinted by $v'$. 
\end{enumerate}
Now, we prove the \emph{only if} part. Suppose that $i \geq 2$ and there exists $v \in N\la u\ra$ such that, for all
$w \in N\la u\ra \setminus \{v\}$, we have $v \noantes w$. We prove that $F$ is not a facet of $P$
by showing that every $(x,y) \in F$ satisfies $y_{vi'} = 0$ for any $i' < i$.
If $x_{ui} = 1$ or $y_{vi} = 1$, then $y_{vi'} = 0$.
If $y_{wi} = 1$ for some $w \in N\la u\ra \setminus \{v\}$, some vertex from $N\la w\ra$ must not be footprinted at step $i-1$ and
then $y_{vi'} = 0$.
\end{proof}
\begin{cor} \label{Prop-RESTR4CLUTTER}
If $G;C$ is a clutter, constraints \eqref{RESTR4} always define facets of $P$.
\end{cor}
\begin{proof}
Straightforward from Lemma \ref{LEMITANV2}.(ii).
\end{proof}

Given a non-empty set of vertices $U \subset V$ and a positive integer $r \leq |U|$, let $N^r\la U\ra$ denote the subset of vertices with exactly $r$ neighbors in $U$, i.e. $$N^r\la U\ra \doteq \{v \in V : |N\la v\ra \cap U| = r\}.$$
Now, we present a new family of valid inequalities. They always dominate constraints \eqref{RESTR2} and, sometimes, constraints \eqref{RESTR4}.

\begin{thm} \label{VERYGENERALINEQ}
Let $i \in \{2,\ldots,m\}$ and $k \in \{1,\ldots,i\}$.
Let $U \subset V$ be a non-empty set with $p$ vertices, $N \subset N^p\la U\ra$ ($N$ possibly empty) and
$W \doteq \{w_1, \ldots, w_t\} \subset N^p\la U\ra \setminus N$ be a non-empty set of vertices, such that:\\
\indent \indent H1) $w_r \noantes w_{r+1}$ for all $r = 1,\ldots,t-1$,\\
\indent \indent H2) $w_t \noantes v$ for all $v \in N$.\\
Let $j_1, \ldots, j_{t+1} \in \{1,\ldots,i\}$ such that $j_1 = 1$, $j_{t+1} = i$ and $j_r \leq j_{r+1}$ for all $r = 1,\ldots,t$.
Then, the following $(i,k,U,N,W,j_1,\ldots,j_{t+1})$-inequality is valid:
\begin{equation} \label{SUPERNOVA}
\sum_{u \in U} x_{ui} + \sum_{v \in N} y_{vi} + \sum_{r=1}^t \sum_{j=j_r}^{j_{r+1}} y_{w_r j} + \sum_{v \in N^p\la U\ra} (p-1) y_{vk} + \sum_{q=1}^{p-1} \sum_{v \in N^q\la U\ra} q y_{vk} \leq p.
\end{equation}
\end{thm}
\begin{proof}
Let $(x,y)$ be a feasible integer point of $P$.
Define
$\Sigma_N \doteq \sum_{v \in N} y_{vi}$,
$\Sigma_W \doteq \sum_{r=1}^t \sum_{j=j_r}^{j_{r+1}}$ $y_{w_r j}$,
$\Sigma_p \doteq \sum_{v \in N^p\la U\ra} (p-1) y_{vk}$,
$\Sigma_q \doteq \sum_{q=1}^{p-1} \sum_{v \in N^q\la U\ra} q y_{vk}$ and
$s \doteq p - \sum_{u \in U} x_{ui}$ ($s$ is the number of vertices $u \in U$ such that $x_{ui} = 0$).
We prove $\Sigma_N + \Sigma_W + \Sigma_p + \Sigma_q \leq s$.

First, note that, if a vertex $v \in N^q\la U \ra$ is chosen in a step from $1,\ldots,i$, then $v$ footprints $q$ vertices from $U$
implying that $1 \leq q \leq s$. Since no more than one vertex can be chosen in step $k$ and $N^1\la U\ra$, $N^2\la U\ra$, $\ldots$,
$N^p\la U\ra$ are disjoint sets of vertices, $\Sigma_q\leq s$ and $\Sigma_p + \Sigma_q \leq p-1$.

Suppose that $s < p$. Then, some vertex from $U$ is not footprinted in step $i$.
Therefore, no vertex from $N^p \la U\ra$ can be chosen at steps $1,\ldots,i$, implying that $\Sigma_N = \Sigma_W = \Sigma_p = 0$.
As $\Sigma_q \leq s$, the inequality is valid.

Now, suppose that $s = p$. Then, all vertices from $U$ are footprinted along steps $1,\ldots,i$. If some vertex from $N$ is chosen at step $i$,
i.e.~$\Sigma_N = 1$, then it is not possible to choose any vertex from $W$ in a step previous to $i$ due to hypotheses H1-H2,
implying that $\Sigma_W = 0$. On the other hand, if $\Sigma_N = 0$ and $y_{w_r j} = 1$ for some $r\in \{1,\ldots,t\}$ and
$j \in \{j_r, \ldots, j_{r+1}\}$, H1 guarantees that it is not possible to choose vertices from $w_1,\ldots,w_{r-1}$ in steps $1,\ldots,j$
nor vertices from $w_{r+1},\ldots,w_t$ in steps $j,\ldots,i$. Therefore, $\Sigma_N + \Sigma_W \leq 1$.
Since $\Sigma_p + \Sigma_q \leq p-1 = s-1$, validity follows.
\end{proof}

When a constraint \eqref{RESTR4} is not facet-defining, it is dominated by an inequality \eqref{SUPERNOVA}. Indeed, let $v'\in N\la u\ra$
such that $v' \noantes w$, for all $w\in N\la u\ra \setminus \{v'\}$. The proof of Prop.~\ref{Prop-RESTR4} suggests that \eqref{RESTR4} can
be strengthened as $x_{ui} + \sum_{v\in N\la u\ra} y_{vi} + \sum_{i'=1}^{i-1} y_{v'i'} \leq 1$.
This is exactly the $(i,i,\{u\},N\la u\ra \setminus \{v'\},\{v'\},1,i)$-inequality.
On the other hand, constraint \eqref{RESTR2} for $v$ is dominated by the $(m,m,\{u\},\emptyset,\{v\},1,m)$-inequality which is
$x_{um}+\sum_{j=1}^m y_{vj} \leq 1$.

Two types of $(i,k,U,N,W,j_1,\ldots,j_{t+1})$-inequalities deserve special attention.
These subfamilies of inequalities are ``naturally'' facet-defining on clutters, as we will see below: 
\begin{itemize}
\item \emph{Type I}: Let $i \in \{2,\ldots,m\}$, $u \in V$ and $w \in N\la u\ra$.
The $(i,i,\{u\},\emptyset,\{w\},1,i)$-inequality is
$$x_{ui} + \sum_{j=1}^i y_{w j} \leq 1.$$
\item \emph{Type II}: Let $i \in \{2,\ldots,m\}$, $k \in \{1,\ldots,i\}$, $u_1, u_2 \in V$ such that $u_1 \neq u_2$,
and $w \in N^{\cap} \doteq N\la u_1\ra \cap N\la u_2\ra$. The $(i,k,\{u_1,u_2\},\emptyset,\{w\},1,i)$-inequality is
$$x_{u_1 i} + x_{u_2 i} + \sum_{j=1}^i y_{wj} + \sum_{v \in N^{\cup}} y_{vk} \leq 2,$$
where $N^{\cup} \doteq N\la u_1\ra \cup N\la u_2\ra$.
Note that the coefficient of $y_{wk}$ is 2 in the l.h.s.
\end{itemize}

To approach Type I inequalities, we first characterize when inequalities \eqref{SUPERNOVA} with $p = 1$ are facet-defining.

\begin{prop}  \label{GENERALINEQ2}
Let $u \in V$ and consider an $(i,i,\{u\},N,W,j_1,\ldots,j_{t+1})$-inequality, i.e.
\begin{equation} \label{NOVA1}
x_{ui} + \sum_{v \in N} y_{vi} + \sum_{r=1}^t \sum_{j=j_r}^{j_{r+1}} y_{w_r j} \leq 1
\end{equation}
satisfying the hypotheses of Theorem \ref{VERYGENERALINEQ}. Then, it defines a facet of $P$ if, and only if:\\
\indent \indent H3) $N\neq \emptyset$ or $N\la w_t\ra \neq \{u\}$ (i.e.$\!$ if $w_t$ is a leaf of $G$ then $w_t \in C$),\\
\indent \indent H4) for every $v \in N\la u\ra \setminus (N \cup W)$, the sets $R^\supset(v)\doteq\{r = 1,\ldots,t : w_r \antes v\}$ and
$R^\subset(v)\doteq\{r = 1,\ldots,t : v \antes w_r\}$ have non-empty intersection, and $j_{r}<j_{r+1}$ for some
$r\in R^\supset(v)\cap R^\subset(v)$.
\end{prop}

\begin{rem} \label{rem:GENERALINEQ2}
We give some useful observations regarding Prop.~\ref{GENERALINEQ2} before proving it:
\begin{enumerate}
	\item \label{it:+1} If $N\neq \emptyset$, say $v \in N$, then condition H2 of Theorem \ref{VERYGENERALINEQ} implies that
	  $N\la w_t\ra \neq \{u\}$ (since $w_t$ and $v$ cannot be twins);
	\item \label{it:ext} By condition H1 of Theorem \ref{VERYGENERALINEQ}, $r\in R^\supset(v)$, $r'\geq r$ imply $r'\in R^\supset(v)$,
	                                 and $r\in R^\subset(v)$, $r'\leq r$ imply $r'\in R^\subset(v)$. 
	      So, for some $s,s'$, we have $R^\supset(v)=\{s,\ldots,t\}$ and $R^\subset(v)=\{1,\ldots,s'\}$, thus giving
				$R^\supset(v)\cap R^\subset(v)=\{s,\ldots,s'\}$ if $s \leq s'$ and $R^\supset(v)\cap R^\subset(v)=\emptyset$ if $s > s'$;
	\item \label{it:max} 
	      $R^\supset(v)\cap R^\subset(v)\neq \emptyset$ is equivalent to ask that the nested sequence
	            $N\la w_1 \ra \supset N\la w_2 \ra \cdots \supset N\la w_t \ra$
	      cannot be extended by $N\la v\ra$:
	      if $R^\supset(v)=\emptyset$, then $v \noantes w_t$;
	      if $R^\subset(v)=\emptyset$, then $w_1 \noantes v$;
	      if $R^\supset(v) \neq \emptyset$, $R^\subset(v) \neq \emptyset$ but $R^\supset(v)\cap R^\subset(v)=\emptyset$, there is $r^*$ such that
				$w_{r^*+1} \noantes v$ and $v \noantes w_{r^*}$;
	\item \label{it:dif} W.l.o.g., we can assume that $j_t<j_{t+1}$. If this is not the case, we do the following. Let $r$ be the largest index
	      such that $j_r < j_{r+1}$. Consider the inequality (\ref{NOVA1})' obtained by using $W'=W\setminus \{w_{r+1},\ldots,w_t\}$,
				$N'=N\cup \{w_{r+1},\ldots,w_t\}$ and $j'_i = j_i$ for all $i = 1,\ldots,r+1$. Note that (\ref{NOVA1}) and (\ref{NOVA1})'
				are equal, and $j'_{t'}<j'_{t'+1}$ where $t' = r$.
\end{enumerate}
\end{rem}

\begin{proof} (of Prop.~\ref{GENERALINEQ2}).
Note that $N \subset N\la u\ra$ and $W \subset N\la u\ra \setminus N$.
The points of the face
$F = \{ (x,y) \in P : x_{ui} + \sum_{v \in N} y_{vi} + \sum_{r=1}^t \sum_{j=j_r}^{j_{r+1}} y_{w_r j} = 1\}$
are characterized as follows:
\emph{a}) $u$ is not footprinted in steps $1,\ldots,i$ or
\emph{b}) some vertex from $N$ is chosen in step $i$ or
\emph{c}) $w_r$ is chosen in some step between $j_r$ and $j_{r+1}$, for some $r\in \{1,\ldots,t\}$.

Let $\beta = \mu^x_{ui}$. Below, we enumerate the conditions needed to prove facet-definition and the points involved to obtain them:
\begin{enumerate}
\item $\mu^x_{u'i'} = 0$ for all $(u',i') \neq (u,i)$.\\
			Let $u'\in V\setminus \{u\}$ and $i'=1,\ldots,m$. The points $(x^1,y^1) = \fp(u',i'; \initial)$, and
			$(x^2,y^2) = \fp(u',i'+1;\initial)$ if $i' < m$ or $(x^2,y^2) = (\initial)$ otherwise, show that $\mu^x_{u'i'} = 0$.
			The same points can be used to show that $\mu^x_{ui'} = 0$ when $i' > i$.
			Now, let $i' < i$. The points $(x^1,y^1) = \ch(w_t,i;\initial)$ and $(x^2,y^2) = \fp(u,i'; x^1,y^1)$ show that
			$\sum_{j=i'}^{i-1} \mu^x_{uj} = 0$.
			By assigning $i'$ in the order $i-1, i-2, \ldots$, we obtain $\mu^x_{u i-1} = 0$, $\mu^x_{u i-2} = 0$, $\ldots$.
			Note that $(x^2, y^2) \in F$ since $w_t$ has a vertex different from $u$ to footprint in step $i$ by hypothesis H3 and
			Remark \ref{rem:GENERALINEQ2}(\ref{it:+1}).
\item $\mu^y_{vi} = \mu^x_{ui}$ for all $v \in N$.\\
			We use points $(x^1,y^1) = \ch(v,i;\initial)$, $(x^2,y^2) = (\initial)$ and apply item 1.
\item $\mu^y_{w_r j} = \mu^x_{ui}$ for all $r = 1,\ldots,t$ and $j \in \{j_r, \ldots, j_{r+1}\}$.\\
			We use points $(x^1,y^1) = \ch(w_r,j;\initial)$, $(x^2,y^2) = (\initial)$ and apply item 1.
\item $\mu^y_{v'i'} = 0$ for all $(v',i')$ such that $v' \in V\setminus N\la u\ra$ or $i' > i$.\\
			Consider $(x^1,y^1) = \ch(v',i';\initial)$, $(x^2,y^2) = (\initial)$ and apply item 1.
\item $\mu^y_{v'i'} = 0$ for all $v' \in N$ and $i' < i$.\\
			Consider $(x^1,y^1) = \ch(w_t,i;\initial)$, $(x^2,y^2) = \ch(v',i';x^1,y^1)$ and apply item 1. Note that $(x^2,y^2) \in F$ because $w_t$ has a
			vertex to footprint in step $i$ by hypothesis H2 of Theorem \ref{VERYGENERALINEQ} and the fact that $v'$ and $w_t$ are not twins.
\item $\mu^y_{v'i'} = 0$ for all $v' \in N\la u\ra \setminus (N\cup W$) and $i' \leq i$.\\
			By hypothesis H4 there is $r$ such that $w_r \antes v'$, $v' \antes w_r$ and $j_{r}<j_{r+1}$. Then, there is $j\in \{j_r,j_{r+1}\}\setminus \{i'\}$. Consider $(x^1,y^1) = \ch(w_r,j;\initial)$ and $(x^2,y^2) = \ch(v',i';x^1,y^1)$. Note that $(x^2,y^2) \in F$ since there is no inclusion relation between $N\la w_r\ra$ and $N\la v'\ra$. Using item 1, we get $\mu^y_{v'i'} = 0$.
\item $\mu^y_{w_r j} = 0$ for all $r = 1,\ldots,t$ and $j \leq i$ such that $j \notin \{j_r, \ldots, j_{r+1}\}$.\\
			Suppose that $j < j_r$. Then, $r \geq 2$. Consider $(x^1,y^1) = \ch(w_{r-1},j_r;\initial)$ and
			$(x^2,y^2) = \ch(w_r,j;x^1, y^1)$. Note that $(x^2,y^2) \in F$ since $N\la w_r\ra \subset N\la w_{r-1} \ra$ and this inclusion is
			strict ($w_r$ and $w_{r-1}$ are not twins) allowing $w_{r-1}$ to be a footprinter at step $j_r$. These points and item 1 lead to $\mu^y_{w_r j} = 0$.
			Now, suppose that $j > j_{r+1}$. Then, $r < t$. A similar argument can be used to obtain the same result with points
			$(x^1,y^1) = \ch(w_{r+1},j_{r+1};\initial)$ and $(x^2,y^2) = \ch(w_r,j;x^1, y^1)$.
\end{enumerate}
For the converse, consider \eqref{NOVA1} after applying Remark \ref{rem:GENERALINEQ2}(\ref{it:dif}).
Then, $j_t < j_{t+1}$. We have the following cases:
\begin{itemize}
\item H3 does not hold, i.e.$\!$ $N=\emptyset$ and $N\la w_t\ra=\{u\}$. Inequality \eqref{NOVA1} is the sum of constraints \eqref{RESTR3}
for $w_t$ and $j_t,\ldots,i-1$, which gives $\sum_{j=j_t+1}^i y_{w_tj} + x_{ui}\leq x_{uj_t}$, together with the $(j_t, j_t, \{u\}, \emptyset, W, j_1, \ldots, j_t, j_t)$-inequality (note that $i$ and $j_{t+1}$ are set to $j_t$).
\item H4 does not hold. Let $v\in N\la u\ra \setminus (N\cup W)$. If $R^\supset(v)=\emptyset$, the stronger 
$(i, i, \{u\},$ $N\cup \{v\}, W, j_1, \ldots, j_{t+1})$-inequality is valid.
Analogously, if $R^\subset(v)=\emptyset$, the stronger $(i, i, \{u\}, N, \{v,w_1,\ldots,w_t\}, j_1, j_1, \ldots, j_{t+1})$-inequality is valid.
If $R^\supset(v)\neq\emptyset$, $R^\subset(v)\neq\emptyset$ but $R^\supset(v)\cap R^\subset(v)=\emptyset$,
Remark \ref{rem:GENERALINEQ2}(\ref{it:max}) ensures the existence of $r^*$ such that the
$(i, i, \{u\}, N,$ $\{w_1,\ldots,w_{r^*},v,w_{r^*+1},\ldots,w_t\},$ $j_1, \ldots , j_{r^*},$ $j_{r^*}, j_{r^*+1}, \ldots, j_{t+1})$-inequality
is a (dominating) valid inequality.
Finally, assume that $R^\supset(v)\cap R^\subset(v)$ is a non-empty subinterval $\{s,\ldots,s'\}$ with $1\leq s\leq s'\leq t$.
Since H4 does not hold, $j_s=j_{s+1}=\ldots=j_{s'}=j_{s'+1}$. 
It follows by Remark \ref{rem:GENERALINEQ2}\eqref{it:ext} that $\{1,\ldots,s-1\}\cap R^\supset(v)=\emptyset$ and $\{s'+1,\ldots,t\}\cap R^\subset(v)=\emptyset$. So, $N\la w_1\ra \supset \cdots \supset N\la w_{s-1}\ra \supset N\la v\ra\supset N\la w_{s'+1}\ra \supset \cdots N\la w_{t}\ra$.
If we add $y_{vj_s}$ to the l.h.s.$\!$ of \eqref{NOVA1}, the inequality remains valid and becomes stronger.
Indeed, by choosing $v$ at step $j_s$, clearly we cannot choose any $w_s$, $w_{s+1}$, $\ldots$, $w_{s'}$ in step $j_s$,  
nor vertices $w_1,\ldots,w_{s-1}$ in steps $1,2,\ldots,j_s$ nor vertices $w_{s'+1},\ldots,w_{t}$ in steps $j_s,\ldots,i$ nor vertices from $N$ in step $i$. 
Besides that, $v$ footprints $u$. 
\end{itemize}
\end{proof}

\begin{cor} \label{FAM1CLUTTER}
If $G;C$ is a clutter, Type I inequalities define facets of $P$.
\end{cor}
\begin{proof}
Using $N=\emptyset$, $t=1$ and $w_t=w$ in Prop.~\ref{GENERALINEQ2},
H3 holds due to Lemma \ref{LEMITANV2}.(ii) and H4 holds
since $R^\supset(v) = R^\subset(v) = \{1\}$ for every $v \in N\la u\ra \setminus \{w\}$.
\end{proof}

We now focus on Type II.
To address each vertex from $\{u_1,u_2\}$ with respect to the other, for $r \in \{1, 2\}$, let $\bar{u_r}$ denote $u_{3-r}$, i.e. $\bar{u_1} = u_2$ and $\bar{u_2} = u_1$.

\begin{prop} \label{FAM2}
A Type II inequality with $k = i$ defines a facet if, and only if:\\
\indent H1) for every $v \in N^{\cap} \setminus \{w\}$, $v \antes w$ and $w \antes v$,\\
\indent H2) for each $r \in \{1, 2\}$, there is $v_r \in N\la u_r\ra \setminus N\la \bar{u_r}\ra$ such that $w \antes v_r$. \\
A Type II inequality with $k=1$ defines a facet if, and only if, H1 holds and:\\
\indent H3) for every $v \in N^{\cup} \setminus N^{\cap}$ such that $w \noantes v$, there exists
$x^v \in (N^{\cup} \setminus N^{\cap}) \setminus \{v\}$ such that 
$|N\la v\ra \cap N\la x^v\ra \cap \{u_1,u_2\}|=1$ and $x^v \antes v$,\\
\indent H4) for each $r \in \{1, 2\}$, there is $z_r \in N\la u_r\ra \setminus N\la \bar{u_r}\ra$ such that $N\la w\ra \setminus (\{\bar{u_r}\} \cup N\la z_r\ra) \neq \emptyset$.\\
A Type II inequality with $1<k<i$ defines a facet if, and only if, H1 and H3 hold and:\\
\indent H5) there exists $v^* \in N^{\cup} \setminus \{w\}$ such that $w \antes v^*$ and $N\la w\ra \setminus (\{u_1,u_2\} \cup N\la v^* \ra)\neq\emptyset$,\\
\indent H6) for each $r \in \{1, 2\}$, one of the following conditions holds:
\begin{itemize}
\item there are $\tilde v_r\in N\la u_r\ra\setminus N\la \bar{u_r}\ra$ and $\tilde v'_r\in N\la\bar{u_r}\ra\setminus \{w\}$ such that $w\antes \tilde v_r$ and $w\antes \tilde v'_r$, or
\item there is $\tilde z_r\in N\la u_r\ra \setminus N\la \bar{u_r}\ra$ such that $N\la w\ra \setminus (\{\bar{u_r}\} \cup N\la \tilde z_r \ra)\neq\emptyset$.
\end{itemize}
\end{prop}

\begin{rem} \label{rem:GENERALINEQ3}
The following observations will be useful in the proof of Prop.~\ref{FAM2}:
\begin{enumerate}
\item If $v \in N^{\cup} \setminus N^{\cap}$ then $v \antes w$ since $w \in N^{\cap}$;
\item In any of the three cases enumerated in Prop.~\ref{FAM2}, $u_1 \antes u_2$ and $u_2 \antes u_1$ (due to H2, H4 or H6);
\item The condition $N\la w\ra \setminus (\{u_1,u_2\} \cup N\la v^* \ra)\neq\emptyset$ in H5 is equivalent to ask
$v^* \antes w$ when $v^*\in N^\cap\setminus \{w\}$, or
$N\la w\ra \setminus (\{\bar{u_r}\} \cup N\la v^* \ra)\neq\emptyset$ when $v^*\in N\la u_r\ra\setminus N\la \bar{u_r}\ra$ for
some $r \in \{1, 2\}$.
\end{enumerate}
\end{rem}
\begin{proof}  (of Prop.~\ref{FAM2}).
The points of the face $F = \{(x,y) \in P : x_{u_1 i} + x_{u_2 i} + \sum_{j=1}^i y_{wj} + \sum_{v \in N^{\cup}} y_{vk} = 2\}$
are characterized as follows:
\emph{a}) neither $u_1$ nor $u_2$ are footprinted in steps $1,\ldots,i$ or
\emph{b}) a vertex from $N\la u_1\ra \setminus N\la u_2\ra$ is chosen in step $k$ and $u_2$ is not footprinted in steps $1,\ldots,i$ or
\emph{c}) the same as (b) with $u_1$ and $u_2$ interchanged or
\emph{d}) $w$ is chosen in a step from $1,\ldots,i$ and a vertex from $N^{\cup}$ (including $w$) is chosen in step $k$.

Let $\beta = \mu^x_{u_1 i}$. Below, we enumerate the conditions needed to prove facet-definition and the points involved to obtain them:
\begin{enumerate}
\item $\mu^x_{u'i'} = 0$ for all $(u',i')$ such that $u'\in V\setminus \{u_1,u_2\}$ or $i' \in \{1,\ldots,m\}\setminus \{k,\ldots,i\}$.\\
			Let $u'\in V$ and $i'=1,\ldots,m$. If $u'\notin \{u_1,u_2\}$ or $i'>i$, these points are in $F$: $(x^1,y^1) = \fp(u',i'; \initial)$ and $(x^2,y^2) = \fp(u',i'+1;\initial)$, if $i' < m$, or $(x^2, y^2) = (\initial)$, otherwise. They show that $\mu^x_{u'i'} = 0$.
			If $u' = u_r$ and $i'<k$, for some $r \in \{1,2\}$, the points $(x^1,y^1) = \ch(w,k; \initial)$ and $(x^2,y^2) = \fp(u',i'; x^1,y^1)$
			lead to $\sum_{j=i'}^{k-1} \mu^x_{u' j} = 0$. By assigning $i'$ in the order $k-1, k-2, \ldots$, we obtain $\mu^x_{u' k-1} = 0$, $\mu^x_{u' k-2} = 0$, $\ldots$
			Note that $(x^2, y^2) \in F$ since $w$ can footprint $\bar{u_r}$ at step $k$.
\item $\mu^x_{\bar{u_r} i'} = 0$ for all $r \in \{1,2\}$ and $k \leq i' < i$.\\
			We consider two cases:  
			\begin{enumerate}
			\item there is $z\in N\la u_r \ra \setminus N\la \bar{u_r} \ra$ such that  $N\la w\ra \setminus (\{\bar{u_r} \} \cup N\la z\ra)\neq\emptyset$. We have that $z \antes w$ and $\bar{u_r} $ can be footprinted before $w$ is chosen. 
			We get $\mu^x_{\bar{u_r}  i'} = 0$ from points $(x^1,y^1) = \ch(w,i'+1;\ch(z,k;\initial))$ and $(x^2,y^2) = \fp(\bar{u_r} ,i'; x^1,y^1)$.  
			\item for all $z\in N\la u_r \ra \setminus N\la \bar{u_r} \ra$, $N\la w\ra \subset \{\bar{u_r} \} \cup N\la z\ra$.
			By H4, $k\geq 2$.
			Consider $\tilde v_r$ and $\tilde v'_r$ from H6 (since the second statement from H6 does not hold).
			Then, the points $(x^1,y^1)=\ch(\tilde v'_r,k;\ch(w,1;\initial))$ and $(x^2,y^2)=\ch(\tilde v_r,k;\ch(w,1;\initial))$ together with item 1 prove that $\mu^y_{\tilde v'_r k}=\mu^y_{\tilde v_r k}$. First, suppose that $N^\cap\neq \{w\}$.
In virtue of H1, we can assume that $\tilde v'_r\in N^\cap\setminus\{w\}$. By H1 and Rem.~\ref{rem:GENERALINEQ3}(1), the points $(x^1,y^1)=\ch(\tilde v'_r,k;\ch(w,i'+1;\initial))$ and $(x^2,y^2)=\ch(\tilde v_r,k;\ch(w,i'+1;\initial))$ belong to $F$
respectively, and then show that 
			\begin{equation}\label{eq:aux-caseb}
			\mu^y_{\tilde v'_r k}=\mu^y_{\tilde v_r k}+\sum_{j=k}^{i'}\mu^x_{\bar{u_r} j}. 
			\end{equation}
			Now, suppose that $N^\cap= \{w\}$. Then, $\tilde v'_r \in N\la \bar{u_r} \ra\setminus N\la u_r \ra$. By Rem.~\ref{rem:GENERALINEQ3}(1), again the same two points are in $F$, but now lead to
			\begin{equation}\label{eq:aux-casebb}
			\mu^y_{\tilde v'_r k}+\sum_{j=1}^{i'}\mu^x_{u_r j}=\mu^y_{\tilde v_r k}+\sum_{j=1}^{i'}\mu^x_{\bar{u_r} j}.
			\end{equation}
			Consider $v^*$ from H5. The condition of case (b) implies that $v^* \notin N\la u_r \ra \setminus N\la \bar{u_r} \ra$ (otherwise, it would contradict H5). As $N^\cap= \{w\}$, $v^*$ must belong to $N\la \bar{u_r} \ra \setminus N\la u_r \ra$.
			From points $(x^1,y^1) = \ch(w,j+1;\ch(v^*,k;\initial))$ and $(x^2,y^2) = \fp(u_r,j; x^1,y^1)$, we get $\mu^x_{u_r j} = 0$ for all $j=k,\ldots,i-1$. Also, by item 1, $\sum_{j=1}^{k-1}\mu^x_{u_r j} = \sum_{j=1}^{k-1}\mu^x_{\bar{u_r} j} = 0$.
Hence, \eqref{eq:aux-casebb} becomes \eqref{eq:aux-caseb}.\\
			Finally, by taking $i'$ in the order $k,k+1,\ldots,i-1$ and using $\mu^y_{\tilde v'_r k}=\mu^y_{\tilde v_r k}$ lead to $\mu^x_{\bar{u_r} i'}=0$ for all $i'=k,\ldots,i-1$.
			\end{enumerate}
\item $\mu^x_{u_2 i} = \mu^x_{u_1 i}$.\\
			Consider $(x^1,y^1) = \ch(v_1,k;\initial)$ and $(x^2,y^2) = \ch(v_2,k;\initial)$ where $v_1 \in N\la u_1\ra \setminus N\la u_2\ra$ and $v_2 \in N\la u_2\ra \setminus N\la u_1\ra$ (they exist due to Rem.~\ref{rem:GENERALINEQ3}(2)).
			After applying items 1-2, we get $\mu^x_{u_2 i} + \mu^y_{v_1 k} = \mu^x_{u_1 i} + \mu^y_{v_2 k}$.
			Hence, it is enough to prove $\mu^y_{v_1 k} = \mu^y_{v_2 k}$.
			If $k < i$, points $(x^1,y^1) = \ch(w,k+1;\ch(v_1,k;\initial))$ and $(x^2,y^2) = \ch(w,k+1;\ch(v_2,k;\initial))$
			belong to $F$ since $v_1 \antes w$ and $v_2 \antes w$ by Rem.~\ref{rem:GENERALINEQ3}(1).
			We get an equality where $\mu^x_{u_1 i}$ and $\mu^x_{u_2 i}$ are absent. Then, we apply items 1-2.
			If $k = i$, we can use points $(x^1,y^1) = \ch(v_1,i;\ch(w,i-1;\initial))$ and $(x^2,y^2) = \ch(v_2,i;\ch(w,i-1;\initial))$,
			which belong to $F$ by H2, and apply item 1.
\item $\mu^y_{v'i'} = 0$ for all $(v',i')$ such that $v' \in V \setminus N^{\cup}$ or $i' > i$.\\
			Consider $(x^1,y^1) = \ch(v',i';\initial)$, $(x^2,y^2) = (\initial)$ and apply item 1.
\item $\mu^y_{vk} = \mu^x_{u_1 i}$ for all $v \in N^{\cup} \setminus N^{\cap}$.\\
			Consider $(x^1,y^1) = \ch(v,k;\initial)$, $(x^2,y^2) = (\initial)$ and apply items 1-2.
			In the case $v \in N\la u_2\ra \setminus N\la u_1\ra$, we obtain $\mu^y_{vk} = \mu^x_{u_2 i}$, so we also apply item 3.
\item $\mu^y_{wk} = 2 \mu^x_{u_1 i}$.\\
			Consider $(x^1,y^1) = \ch(w,k;\initial)$, $(x^2,y^2) = (\initial)$ and apply items 1, 2 and 3.
\item $\mu^y_{wj} = \mu^x_{u_1 i}$ for all $j \in \{1,\ldots,i\} \setminus \{k\}$.\\
			Consider the point $(x^1,y^1) = \ch(v,k;\ch(w,j;\initial))$ where $v$ is given as follows:
			\begin{itemize}
			\item Case $k = i$: Here $j < i$. Let $v = v_2$ be the vertex given by H2. Then, $w \antes v_2$.
			\item Case $j < k < i$: Consider the vertex $v = v^*$ from H5 satisfying $w \antes v^*$.
			\item Case $k < j \leq i$: Consider some $v \in N^{\cup} \setminus N^{\cap}$ by Rem.~\ref{rem:GENERALINEQ3}(2).
			Note that $v \antes w$ by Rem.~\ref{rem:GENERALINEQ3}(1).
			\end{itemize}
			Take $(x^2,y^2) = (\initial)$ and apply items 1, 2 and 3.
			We get $\mu^y_{wj} + \mu^y_{vk} = 2 \mu^x_{u_1 i}$. Then, apply item 5.
\item $\mu^y_{vk} = \mu^x_{u_1 i}$ for all $v \in N^{\cap} \setminus \{w\}$.\\
			Take $j \in \{1,\ldots,i\} \setminus \{k\}$. H1 guarantees that $v$ and $w$ can be chosen at steps $k$ and $j$
			respectively, and the points $(x^1,y^1) = \ch(v,k;\ch(w,j;\initial))$ and $(x^2,y^2) = \ch(w,k;\initial)$ show that
			$\mu^y_{vk} + \mu^y_{wj} = \mu^y_{wk}$ by items 1-2. Then, apply items 6 and 7.
\item $\mu^y_{v'i'} = 0$ for all $v' \in N^{\cup} \setminus \{w\}$ and $i' \in \{1,\ldots,i\} \setminus \{k\}$.\\
			First, consider the case $v' \in N^{\cup} \setminus N^{\cap}$, $i' > k$ and $w \noantes v'$. In virtue of H3, there exists
			$x^{v'}$ such that $x^{v'} \antes v'$ and $N\la v'\ra \cap N\la x^{v'}\ra \cap \{u_r,\bar{u_r}\} = \{u_r\}$ for some $r$.
			Hence, neither $v'$ nor $x^{v'}$ is adjacent to $\bar{u_r}$.
			Therefore, $v'$ and $x^{v'}$ can be chosen at steps $i'$ and $k$ respectively, and $\bar{u_r}$ is not footprinted in $1,\ldots,i$.
			The points $(x^1,y^1) = \ch(x^{v'},k;\initial)$ and $(x^2,y^2) = \ch(v',i';x^1,y^1)$ lead to $\mu^y_{v'i'} = 0$ by items 1-2.
			For any other case, the points $(x^1,y^1) = \ch(w,k;\initial)$ and $(x^2,y^2) = \ch(v',i';x^1,y^1)$ lead to $\mu^y_{v'i'} = 0$
			again by items 1-2.
			It only remains to prove that $(x^2,y^2)$ lies on $F$.
			If $v' \in N^{\cap} \setminus \{w\}$, H1 ensures that $v'$ and $w$ can be chosen at steps $i'$ and $k$ respectively.
			Now consider $v' \in N^{\cup} \setminus N^{\cap}$.
			If $i' < k$, note that $v' \antes w$ by Rem.~\ref{rem:GENERALINEQ3}(1), so we are done.
			Now consider $i' > k$. Therefore, $w \antes v'$ and, again, $(x^2,y^2) \in F$.
\end{enumerate}
We now prove the converse.
If H1 does not hold, there is $v \in N^{\cap} \setminus \{w\}$ such that $v \noantes w$ or $w \noantes v$. In the first case,
any solution $(x,y) \in F$ satisfies $y_{v1} = 0$.
Similarly, in the case that $w \noantes v$, any solution $(x,y) \in F$ satisfies $y_{vi} = 0$.

To prove necessity for the other hypotheses, we use the following auxiliary result:
\begin{claim}\label{claim}
Suppose that, for some $r \in \{1,2\}$, every $v \in N\la u_r\ra \setminus N\la \bar{u_r}\ra$ satisfies:
(i) $w \noantes v$, if $k=i$;
(ii) $N\la w\ra \subset \{\bar{u_r}\}\cup N\la v\ra$, if $k=1$; or
(iii) $w \noantes v$ and $N\la w\ra \subset \{\bar{u_r}\}\cup N\la v\ra$, if $1<k<i$.
Then, every $(x,y)\in F$ satisfies constraint \eqref{RESTR4} for $\bar{u_r}$ and $k$ as equality,
i.e.~$x_{\bar{u_r} k} + \sum_{z \in N\la \bar{u_r}\ra} y_{zk} = 1$.
\end{claim}
\begin{proof}
Let $(x,y)\in F$.
Suppose by contradiction that $x_{\bar{u_r} k} + \sum_{z \in N\la \bar{u_r}\ra} y_{zk} = 0$. Then $x_{\bar{u_r} i}=0$ (due to $x_{\bar{u_r} k}=0)$ and $y_{wk}=0$ (due to $y_{zk}=0$ for all $z\in N\la \bar{u_r}\ra)$.
Therefore, as $(x,y) \in F$,
$$x_{u_r i} + \sum_{\substack{j=1\\j \neq k}}^i y_{wj} + \sum_{v \in N\la u_r\ra \setminus N\la \bar{u_r}\ra} y_{vk} = 2.$$
If $x_{u_r i}=1$, both summations vanish which is absurd. Then, $x_{u_r i}=0$,
some vertex $v \in N\la u_r\ra \setminus N\la \bar{u_r}\ra$ is chosen at step $k$ and $w$ is chosen at step $j\in \{1,\ldots,i\} \setminus \{k\}$. Since $x_{\bar{u_r} k}=0$, it follows that $w\antes v$ (if $j<k$) or $N\la w\ra \setminus (\{\bar{u_r}\}\cup N\la v\ra)\neq \emptyset$ (if $j>k$). In any case, we get a contradition. 
\end{proof}	

Consider that $k=i$. Assume that H2 does not hold, i.e.~for some $r$, every $v_r \in N\la u_r\ra \setminus N\la \bar{u_r}\ra$ satisfies
$w \noantes \bar{u_r}$. Claim \ref{claim}(i) implies that $F$ is not a facet.

If $k=1$ and H4 does not hold, i.e.~for some $r$, every $z_r \in N\la u_r\ra \setminus N\la \bar{u_r}\ra$ satisfies $N\la w\ra \subset \{\bar{u_r}\}\cup N\la z_r\ra$, we obtain the same conclusion by Claim \ref{claim}(ii).

Consider now that $k<i$. If H3 does not hold, for some $r$ there is $v \in N\la u_r\ra \setminus N\la \bar{u_r}\ra$ such that
$w \noantes v$ and for all $x^v \in (N\la u_r\ra \setminus N\la \bar{u_r}\ra) \setminus \{v\}$, $x^v \noantes v$.
Here, any solution that chooses $v$ in step $i$ does not belong to $F$, implying that $y_{vi}=0$.

Finally, consider that $1<k<i$. Suppose that H5 is not valid, i.e.~for all $v^* \in N^{\cup} \setminus \{w\}$, $w \noantes v^*$ or
$N\la w\ra \subset N\la v^* \ra \cup \{u_1,u_2\}$. In the first case, i.e.~$w \noantes v^*$, the stronger inequality
$x_{u_1 i} + x_{u_2 i} + \sum_{j=1}^{k-1} 2 y_{wj} + \sum_{j=k}^i y_{wj} + \sum_{v \in N^{\cup}} y_{vk} \leq 2$ is valid.
In the second case, we claim that every $(x,y)\in F$ satisfies at equality the Type II inequality
\begin{equation} \label{eq:auxk=i}
x_{u_1 k} + x_{u_2 k} + \sum_{j=1}^k y_{wj} + \sum_{v\in N^\cup} y_{vk} \leq 2
\end{equation}
In other words, we claim that $(x_{u_1 k}- x_{u_1 i}) + (x_{u_2 k}- x_{u_2 i}) - \sum_{j=k+1}^i y_{wj}=0$. Actually, by the validity of \eqref{eq:auxk=i}, it is enough to prove that $(x_{u_1 k}- x_{u_1 i}) + (x_{u_2 k}- x_{u_2 i}) - \sum_{j=k+1}^i y_{wj}\geq 0$. If $\sum_{j=k+1}^i y_{wj}= 0$, we are done because $x_{uk}\geq x_{ui}$ for all $u\in V$. Otherwise, $w$ was chosen in some step $j\in \{k+1,\ldots,i\}$. So, $x_{u_1 i}=x_{u_2 i}=0$. As $(x,y)\in F$, a vertex $v\in N^\cup\setminus \{w\}$ must be chosen at step $k$, implying $x_{u_1 k} +x_{u_2 k}\leq 1$. Since $N\la w\ra \subset N\la v\ra \cup \{u_1,u_2\}$ by assumption, either $u_1$ or $u_2$ was not footprinted at step $j$, that is, $x_{u_1 j} + x_{u_2 j}\geq 1$, or still $x_{u_1 k} +x_{u_2 k}\geq 1$. It follows that $(x_{u_1 k}- x_{u_1 i}) + (x_{u_2 k}- x_{u_2 i})=1$, as desired.

To complete the converse, suppose that H6 does not hold. On the one hand, $N\la w\ra \subset \{\bar{u_r}\} \cup N\la \tilde z_r \ra$ for all $\tilde z_r\in N\la u_r\ra \setminus N\la \bar{u_r}\ra$. On the other hand, either $w\noantes \tilde v_r$ for all $\tilde v_r\in N\la u_r\ra\setminus N\la \bar{u_r}\ra$ or $w\noantes \tilde v'_r$ for all $\tilde v'_r\in N\la \bar{u_r}\ra\setminus \{w\}$.
By Claim~\ref{claim}(iii), it remains to prove the second case, i.e.~$w\noantes \tilde v'_r$ for all $\tilde v'_r\in N\la \bar{u_r}\ra\setminus \{w\}$. This latter condition ensures that the $(k,k,\{\bar{u_r}\},N\la \bar{u_r}\ra \setminus \{w\},\{w\},1,k)$-inequality is valid
(see H2 of Theorem~\ref{VERYGENERALINEQ}): $x_{\bar{u_r} k} + \sum_{v\in N\la \bar{u_r}\ra} y_{vk} + \sum_{j=1}^{k-1} y_{wj} \leq 1$. We show that it is satisfied at equality by all points in $F$. Suppose by contradiction that $x_{\bar{u_r} k} + \sum_{v\in N\la \bar{u_r}\ra} y_{vk} + \sum_{j=1}^{k-1} y_{wj} =0$. As $(x,y)\in F$, some vertex $\tilde z_r \in N\la u_r\ra \setminus N\la \bar{u_r}\ra$ is chosen at step $k$ and $w$ is chosen at step $j\in \{k+1,\ldots,i\}$. Since $x_{\bar{u_r} k}=0$, it follows that $N\la w\ra \setminus (\{\bar{u_r}\}\cup N\la \tilde z_r\ra)\neq \emptyset$: a contradiction.
Therefore, $F$ is not a facet.
\end{proof}

\begin{cor} \label{FAM2CLUTTER}
Type II inequalities define facets if $G;C$ is a clutter and $k=i$, or $G;C$ is a strong clutter.
\end{cor}

The following table summarizes the facet-defining results for constraints of $F_1$ and the new families of valid inequalities,
and mentions if a stronger inequality was found. A dagger ``$^\dagger$'' indicates whether no extra conditions are needed for the inequality
to define a facet:

\begin{center}
\begin{tabular}{|c|c|c|c|}
\hline
    & \multicolumn{2}{c|}{¿Is $G;C$ a clutter?} & \\
Constraint & No & Yes & Stronger \\
\hline
Non-neg.          & \multicolumn{2}{c|}{Prop.~\ref{EASYFACETS} (i)$^\dagger$, (ii)$^\dagger$}               & \\
\eqref{RESTR1}    & Prop.~\ref{RESTR1STRONGER}         & Cor.~\ref{RESTR1STRONGERCLUTTER}$^\dagger$ & \eqref{RESTR1STRONGER1}, \eqref{RESTR1STRONGER2}, \eqref{RESTR1STRONGER3} \\
\eqref{RESTR2}    & \multicolumn{2}{c|}{not a facet}                                              & \eqref{NOVA1}\\ 
\eqref{RESTR3}    & Prop.~\ref{GENERALINEQ1}         & Cor.~\ref{GENERALINEQ1CLUTTER}$^\dagger$   & \eqref{NOVA0} \\
\eqref{RESTR4}    & Prop.~\ref{Prop-RESTR4}          & Cor.~\ref{Prop-RESTR4CLUTTER}$^\dagger$    & \eqref{NOVA1} \\
\eqref{RESTR5}    & Prop.~\ref{EASYFACETS} (iii)     & Cor.~\ref{EASYFACETSCLUTTER}$^\dagger$     & \\
Type I            & Prop.~\ref{GENERALINEQ2}         & Cor.~\ref{FAM1CLUTTER}$^\dagger$           & \eqref{NOVA1} \\
Type II($k=i$)    & Prop.~\ref{FAM2}                 & Cor.~\ref{FAM2CLUTTER}$^\dagger$           & \eqref{SUPERNOVA} \\
Type II($k<i$)    & Prop.~\ref{FAM2}                 & Cor.~\ref{FAM2CLUTTER}(only strong)$^\dagger$           & \eqref{SUPERNOVA} \\
\hline
\end{tabular}
\end{center}

\subsection{Twin vertices and facet generation}

Here, we explore what happens if twin vertices are introduced.
Let $G$ be a connected graph with at least 3 vertices. Assume $m\geq 3$.
Given a valid inequality $\pi^x x + \pi^y y\leq \pi_0$ for $P$ and a vertex $u\in V$, let $\pi^x_u$ denote the subvector of $\pi^x$ related to $u$, i.e. $\pi^x_u=(\pi^x_{u1},\ldots,\pi^x_{um})$. Similarly, define $\pi^y_u$, $x_u$ and $y_u$ with respect to $\pi^y$, $x$ and $y$, respectively.

\begin{lemma} \label{lemma:twins}
Let $u, u'$ be twin vertices of $G;C$ and $\pi^x x + \pi^y y\leq \pi_0$ be facet-defining for $P$. Then:\\
(i)~$\pi^y_u=\pi^y_{u'}$, if the inequality is different from $-y_{uj}\leq 0$ and $-y_{u'j}\leq 0$, for all $j=1,\ldots,m$;\\
(ii)~$\pi^x_{u'} x_{u} + \pi^x_{u} x_{u'} + \pi^x_* x_* + \pi^y y \leq \pi_0$ is facet-defining for $P$, where $\pi^x_*$ and $x_*$ are vectors $\pi^x$ and $x$ without components $(\pi^x_u,\pi^x_{u'})$ and $(x_u,x_{u'})$, respectively;\\
(iii)~if $\pi^x_{uk}\geq 0$ and $\pi^x_{u'k}\geq 0$ for all $k=1,\ldots,j$, then $\pi^x_{uj}=0$ or $\pi^x_{u'j}=0$, for all $j=1,\ldots,m$.
\end{lemma}
\begin{proof} Let $F=\{(x,y)\in P:\pi^x x + \pi^y y = \pi_0\}$.\\ 
(i)~Since the inequality is different from $-y_{u'j}\leq 0$, there is a point $(x^1,y^1) \in F$ such that $y^1_{u'j}=1$.
Then, $y^1_{uk}=0$ for all $k=1,\ldots,m$.
Let $y^2$ be the vector obtained from $y^1$ by swapping components from $u$ and $u'$: $y^2_{uk} = y^1_{u'k}$ and
$y^2_{u'k} = y^1_{uk}$ for $k = 1,\ldots,m$.
Note that $(x^1,y^2)\in P$ because $N\la u\ra=N\la u'\ra$. Since the inequality is valid, $\pi^x x^1 + \pi^y y^2\leq \pi_0$.
Also, $\pi^x x^1 + \pi^y y^1= \pi_0$, which implies $\pi^y(y^2-y^1)= \pi^y_{uj} - \pi^y_{u'j} \leq 0$.
A similar reasoning swapping the roles of $u$ and $u'$ leads to $\pi^y_{u'j}-\pi^y_{uj}\leq 0$. Therefore, $\pi^y_{uj}=\pi^y_{u'j}$.\\
(ii)~Let $(x^1_u,x^1_{u'},x^1_*,y^1)\in P$. Observe that the point $(x^2,y^1)$ such that $x^2_u=x^1_{u'}$, $x^2_{u'}=x^1_u$ and $x^2_*=x^1_*$ is also in $P$.
As the inequality is valid, $\pi^x x^2 + \pi^y y^1\leq \pi_0$. It follows that $\pi^x_{u'} x^1_{u} + \pi^x_{u} x^1_{u'} + \pi^x_* x^1_* + \pi^y y^1 = \pi^x_{u'} x^2_{u'} + \pi^x_{u} x^2_{u} + \pi^x_* x^2_* + \pi^y y^1 \leq  \pi_0$. Hence, $\pi^x_{u'} x_{u} + \pi^x_{u} x_{u'} + \pi^x_* x_* + \pi^y y \leq \pi_0$ is valid too. Now, let $F'$ be the face defined by the latter inequality.
Pick $2mn$ affinely independent points in $F$. By interchanging the values of the components $x_u$ and $x_{u'}$ as before, we get $2mn$ affinely independent points in $F'$. Therefore, $F'$ is also a facet of $P$.\\
(iii) As $F$ is a facet, $F \not\subset\{x\in P:x_{uj}=x_{u'j}\}$ and, thus, there is an integer point $(x^1,y^1)\in F$ such that $x^1_{uj}\neq x^1_{u'j}$. Consider the case $x^1_{uj}=1$ and $x^1_{u'j}=0$. Since $u$ is not footprinted at step $j$, we can obtain another point in $P$ where $u'$ is also not footprinted at step $j$. Precisely, $(x^2,y^1)\in P$ where $x^2_{u'k}=1$, for all $k=1,\ldots,j$, 
and the other components of $x^2$ coincide with those of $x^1$. Therefore, $\pi^x x^2 + \pi^y y^1\leq \pi_0=\pi^x x^1 + \pi^y y^1$, which implies $0\geq \pi^x x^2 - \pi^x x^1 = \sum_{k=1}^{j}\pi^x_{u'k}(x^2_{u'k}-x^1_{u'k}) = \sum_{k=1}^{j-1}\pi^x_{u'k}(1-x^1_{u'k}) + \pi^x_{u'j}\geq \pi^x_{u'j}$
since $\pi^x_{u'k} \geq 0$ for $k = 1,\ldots,j-1$. By $\pi^x_{u'j} \geq 0$, we get $\pi^x_{u'j}=0$.
Analogously, the case $x^1_{uj}=0$ and $x^1_{u'j}=1$ leads to $\pi^x_{uj}=0$.
\end{proof}

\begin{thm} \label{prop:lift-twins}
Let $G';C'$ be an instance where $u, u'$ are twin vertices, $G = G' - u'$ (i.e.$\!$ the graph that results from
deleting $u'$ to $G'$) and $C = C'\setminus\{u'\}$. Let $P$ and $P'$ be the polytopes corresponding to $G;C$ and $G';C'$ respectively (using $m$ as upper bound). If the inequality $\pi^x x + \pi^y y\leq \pi_0$ is valid 
for $P$ and $\pi^x_u\geq 0$, then
\begin{align}
\pi^x_* x_* + \pi^x_{u} x_{u}\; + \pi^y y + \pi^y_u y_{u'}& \leq  \pi_0, \label{lift1}\\
\pi^x_* x_* + \pi^x_{u} x_{u'} + \pi^y y + \pi^y_u y_{u'}& \leq  \pi_0 \label{lift2}
\end{align}
are valid for $P'$, where $\pi^x_*$ and $x_*$ are vectors $\pi^x$ and $x$ without components $\pi^x_u$ and $x_u$, respectively.
Moreover, if $\pi^x x + \pi^y y\leq \pi_0$ is facet-defining for $P$ and $\pi_0\neq 0$, then \eqref{lift1}-\eqref{lift2} are facet-defining for $P'$.
\end{thm}
\begin{proof}
Let $\pi^y_*$ and $y_*$ denote vectors $\pi^y$ and $y$ without components $\pi^y_u$ and $y_u$, respectively.
Let $(x_u,x_{u'},x_{*},y_u,y_{u'},y_{*})\in P'$. Define $x^{\max}_{u}$ as the maximum between $x_u$ and $x_{u'}$, where the values are taken componentwise: $x^{\max}_{uj} = \max\{x_{uj},x_{u'j}\}$ for all $j$.
Since either $u$ or $u'$ can not be chosen simultaneously, we have $y_u=\bf 0$ or $y_{u'}=\bf 0$.
First, consider the case $y_{u'}=\bf 0$. Define $y \doteq (y_u,y_*)$.
It is easy to check that $(x^{\max}_u,x_{*},y)$ satisfies \eqref{RESTR1}-\eqref{RESTR5} so $(x^{\max}_u,x_{*},y) \in P$.
Therefore, $\pi^x_* x_* + \pi^x_{u} x^{\max}_{u} + \pi^y y \leq \pi_0$.
As $\pi^x_u\geq 0$ and $y_{u'}=\bf 0$, we get $\pi^x_* x_* + \pi^x_{u} x_{u} + \pi^y y + \pi^y_u y_{u'}\leq \pi^x_* x_* + \pi^x_{u} x^{\max}_{u} + \pi^y y \leq \pi_0$.
This shows that \eqref{lift1} is valid.
Similarly, it can be proven that \eqref{lift1} is valid in the case $y_u=\bf 0$ by defining $y \doteq (y_{u'},y_*)$.

Next, we prove that \eqref{lift1} is facet-defining.
Let $F = \{(x,y) \in P : \pi^x x + \pi^y y= \pi_0\}$. As $\pi_0\neq 0$, the origin is in $P\setminus F$, and so it
does not belong to the affine hull of $F$.
Consequently, there are $2mn$ linearly independent vectors in $F$, where $n=|V(G)|$.
Let $A=[X_u\: Y_u\: X_*\: Y_* ]$ be the square matrix where each row is one of these vectors, and the columns are grouped according to the components $x_u$, $y_u$, $x_*$ and $y_*$. Since $A$ is invertible, there is an $m\times m$ invertible submatrix $X^1_u$ of $X_u$. In addition, there is another $m\times m$ invertible submatrix $Y^2_u$ of $Y_u$.

Let $Y^1_u$, $X^1_*$ and $Y^1_*$ be the submatrices of $Y_u$, $X_*$, $Y_*$ using the same rows as $X^1_u$, and $X^2_u$,
$X^2_*$ and $Y^2_*$ be the submatrices of $X_u$, $X_*$, $Y_*$ using the same rows as $Y^2_u$.
Define the $2m(n+1)\times 2m(n+1)$ matrix
\[B=\begin{bmatrix}
X_u   & 0     & Y_u   & 0     & X_*   & Y_*\\
X^1_u & X^1_u & Y^1_u & 0     & X^1_* & Y^1_*\\
X^2_u & 0     & 0     & Y^2_u & X^2_* & Y^2_*\\
\end{bmatrix},\]
where each row has the form $(x_u,x_{u'},y_u,y_{u'},x_*,y_*)$.
We now prove that any row belongs to $P'$. Suppose w.l.o.g.~that $A$ is a $\{0,1\}$-matrix, then the rows correspond to legal sequences of $G;C$.
In the first group of $2mn$ rows, we take a point in $P$, do not choose $u'$ but footprint $u'$ since step 1. In the second group of $m$ rows, we do not choose $u'$ but footprint $u'$ exactly as $u$ was footprinted, which is possible since they are twins. In the last group of $m$ rows, we footprint $u'$ since step 1 and, in the case that $u$ was chosen in the starting point taken from $P$ at, say, step $j$, then choose $u'$ at $j$ and do not choose $u$.

Since the first group of $2mn$ rows is linearly independent, and $X^1_u$ and $Y^2_u$ are invertible, $B$ is invertible too. Moreover, each row $(x_u,x_{u'},y_u,y_{u'},x_*,y_*)$ is a point in the face defined by \eqref{lift1}.
Indeed, as $\pi^x_u x_u + \pi^y_u y_u + \pi^x_* x_* + \pi^y_* y_* = \pi_0$, we readily get
$\pi^x_u x_u + 0 x_{u'} + \pi^y_u y_u + \pi^y_u y_{u'} + \pi^x_* x_* + \pi^y_* y_* = \pi_0$. This means that each row-vector satisfies \eqref{lift1} at equality. Thus, we obtain $2m(n+1)$ linearly independent points in the face defined by \eqref{lift1}
implying that it is a facet of $P'$.

The results concerning~\eqref{lift2} are a direct consequence of those for \eqref{lift1} and Lemma~\ref{lemma:twins}(ii).
\end{proof}

\section{Computational experiments} \label{SSCOMPU}

This section is devoted to present computational experiments in order to know which formulation performs better
and if the addition of cuts coming from the previous polyhedral study is effective during the optimization.

The experiments have been carried out by running a pure Branch and Bound over several random instances.
A computer equipped with an Intel i5 CPU 2.67GHz, 4Gb of RAM, Ubuntu 16.04 operating system and IBM ILOG CPLEX 12.7
has been used. Each run has been performed on one thread of the CPU with a limit of two hours of time. 

Random instances are generated as follows. For given numbers $n \in \mathbb{Z}_+, p \in [0, 1]$, a graph is generated by
starting from the empty graph of $n$ vertices and adding edges with probability $p$. For example, if $p = 1$ then a $K_n$ is obtained.

For each combination $(n,p) \in \{(15,0.2)$, $(15,0.4)$, $(15,0.6)$, $(15,0.8)$, $(20,0.4)$, $(20,0.6)$, $(20,0.8)$, $(25,0.6)$,
$(25,0.8)$, $(30,0.8)\}$ and each $C \in \{ \emptyset, V, \frac{V}{2} \}$ (total GDP, classic GDP and a case of GGDP where $C$ contains
half of the vertices), 5 graphs are generated, thus giving a total of 150 instances. All of them are connected and twin free.
Combinations with a large amount of nodes and a low density (such as $n=20$ and $p=0.2$) were deliberately discarded since almost all
instances with those combinations can not be solved in the preset time limit. In particular, these instances present a high value of the
initial upper bound $m$.

\subsection{Initial heuristic}

In order to provide a reasonably good lower bound, an initial solution is constructed and injected to CPLEX at the beginning of the
optimization.

We implemented a greedy algorithm to construct it. Start from empty sets $S$ and $W$, and assign $(x,y) \leftarrow (\initial)$.
At each step $i$, pick a vertex $v \in V \setminus S$ such that the cardinal of $N\la v\ra \setminus W$ is minimum but not zero,
then add $v$ to $S$, $N\la v\ra$ to $W$ and do $(x,y) \leftarrow \ch(v,i;x,y)$.  Repeat until $W = V$.

The resulting solution $(x,y)$ represents a legal dominating sequence of maximal length and is feasible for the eight formulations.
In addition, $LB$ is set to the length of that sequence.

\subsection{Comparing formulations}

All the formulations are evaluated over the instances generated from the 7 combinations $(15,0.2)$, $(15,0.4)$, $(15,0.6)$, $(20,0.6)$,
$(20,0.8)$, $(25,0.8)$, $(30,0.8)$. Combinations $(20,0.4)$ and $(25,0.6)$ are not considered since they are too difficult for some
formulations. On the other hand, $(15,0.8)$ is too easy: any formulation requires a few seconds to solve them.

Table \ref{tab:1} reports the results. The first three columns give the combination $(n,p,C)$ and the fourth one
indicates the parameters being evaluated: ``\#i(\% relgap)'' refers to the number of solved instances and the average of the
percentage relative gap, ``Nodes'' and ``Time'' are the averages of the number of nodes evaluated by the B\&B and
the elapsed time (in seconds), respectively, over the solved instances. For $(n,p) \in \{(15,0.4)$, $(15,0.6)$, $(20,0.8)$, $(25,0.8)\}$,
``\#i(\% relgap)'' is not reported since all the formulations are able to solve the instances of the given combination.

The last four rows of the table summarize the results for all combinations, with best values in boldface.
In particular, the last row (``Time$^*$'') reports the average of the elapsed time over all 105 instances considered in this experiment, by
taking 7200 sec. for those that are not solved. This parameter can also be computed as follows:
$$ \textrm{Time}^* = \dfrac{\textrm{Time} \times \textrm{\#i} + 7200(105 - \textrm{\#i})}{105}.$$

As we can see from the table, $F_3$, $F_7$ and $F_8$ solve all the instances. In particular, $F_8$ explores less nodes than the others.
However, the linear relaxations of $F_3$ are easier to solve and, in consequence, the CPU time is smaller. According to these conclusions,
the best formulation is $F_3$.

We also measure the impact of individually adding the extra constraints to the initial formulation by means of the parameter Time$^*$.
Let $T_i$ be the value of Time$^*$ for $F_i$. We first note that by considering constraints \eqref{RESTR8}-\eqref{RESTR9} reduces
dramatically the overall CPU time: we compare $T_1 + T_2 + T_5 + T_6$ (those constraints are absent) with $T_3 + T_4 + T_7 + T_8$
(those constraints are present), which give 6258 and 1184 seconds respectively.
Constraints \eqref{RESTR10} also help to diminish the CPU time but the improvement is marginal:
$T_1 + T_2 + T_3 + T_4 = 3821$ and $T_5 + T_6 + T_7 + T_8 = 3621$.
On the other hand, constraints \eqref{RESTR6}-\eqref{RESTR7} make the optimizer behave slightly worse:
$T_1 + T_3 + T_5 + T_7 = 3674$ and $T_2 + T_4 + T_6 + T_8 = 3768$.

Finally, we compare the difficulty of the classic GDP vs.$\!$ total GDP with the data gathered in this experiment.
For a given $C \in \{ \emptyset, V, \frac{V}{2} \}$, let $T_C$ be the average of time over all the formulations and all instances
associated with that $C$.
We obtain $T_{\emptyset} = 2511$, $T_{V} = 2171$ and $T_{\frac{V}{2}} = 2760$ which suggests that total GDP is harder than the classic one.
A cause for this behavior could be the following: for a given graph $G$, the upper bound of $\grd(G;\emptyset)$ has one more unit
than the upper bound of $\grd(G;V)$, implying that the formulation (for the same graph) requires $2|V|$ additional variables in the total
domination case.

\begin{table}[t] \centering \footnotesize
\begin{tabular}{|@{\hspace{3pt}}c@{\hspace{3pt}}|@{\hspace{3pt}}c@{\hspace{3pt}}|@{\hspace{3pt}}c@{\hspace{3pt}}|@{\hspace{3pt}}c@{\hspace{3pt}}|@{\hspace{3pt}}c@{\hspace{3pt}}|@{\hspace{3pt}}c@{\hspace{3pt}}|@{\hspace{3pt}}c@{\hspace{3pt}}|@{\hspace{3pt}}c@{\hspace{3pt}}|@{\hspace{3pt}}c@{\hspace{3pt}}|@{\hspace{3pt}}c@{\hspace{3pt}}|@{\hspace{3pt}}c@{\hspace{3pt}}|@{\hspace{3pt}}c@{\hspace{3pt}}|}
\hline
 $n$ & $p$ & $C$ &  &  $F_1$ & $F_2$ & $F_3$ & $F_4$ & $F_5$ & $F_6$ & $F_7$ & $F_8$ \\
\hline
    &     &               & \#i(\% relgap) & 5(0.00) & 5(0.00) & 5(0.00) & 5(0.00) & 5(0.00) & 5(0.00) & 5(0.00) & 5(0.00)\\ 
    &     & $\emptyset$   & Nodes & 92923 & 77751 & 43477 & 61749 & 94655 & 49019 & 64102 & 48598\\ 
    &     &                & Time & 175.9 & 172.7 & 113.9 & 227.0 & 198.9 & 119.7 & 159.1 & 151.7\\ 
    &     &               & \#i(\% relgap) & 5(0.00) & 5(0.00) & 5(0.00) & 5(0.00) & 5(0.00) & 5(0.00) & 5(0.00) & 5(0.00)\\ 
 15 & 0.2 & $V$           & Nodes & 2495260 & 1418036 & 289997 & 96119 & 1789791 & 933478 & 340730 & 179896\\ 
    &     &                & Time & 2738.1 & 2000.6 & 518.6 & 258.9 & 2284.6 & 1548.6 & 693.3 & 527.9\\ 
    &     &               & \#i(\% relgap) & 4(1.85) & 5(0.00) & 5(0.00) & 5(0.00) & 4(1.82) & 5(0.00) & 5(0.00) & 5(0.00)\\ 
    &     & $\frac{V}{2}$ & Nodes & 1364835 & 1773348 & 325082 & 350817 & 1459272 & 1381170 & 379803 & 133720\\ 
    &     &                & Time & 1272.9 & 2301.3 & 448.2 & 673.6 & 1629.7 & 1791.0 & 582.6 & 309.8\\ 
\hline
    &     & $\emptyset$   & Nodes & 712710 & 286807 & 48922 & 73541 & 608650 & 421643 & 66406 & 39945\\ 
    &     &                & Time & 564.4 & 325.4 & 101.3 & 212.5 & 471.0 & 500.9 & 174.0 & 113.5\\ 
 15 & 0.4 & $V$           & Nodes & 1268312 & 544577 & 70198 & 19714 & 1014941 & 692413 & 60680 & 29948\\ 
    &     &                & Time & 1115.4 & 717.0 & 144.3 & 54.4 & 904.0 & 852.8 & 135.9 & 81.7\\ 
    &     & $\frac{V}{2}$ & Nodes & 602265 & 513324 & 35499 & 29511 & 235344 & 165470 & 44559 & 28634\\ 
    &     &                & Time & 709.5 & 1018.5 & 94.1 & 91.5 & 258.9 & 267.4 & 98.8 & 77.5\\ 
\hline
    &     & $\emptyset$   & Nodes & 30088 & 22556 & 4703 & 4271 & 23614 & 34649 & 5076 & 4668\\ 
    &     &                & Time & 19.9 & 25.6 & 6.6 & 9.2 & 17.9 & 38.9 & 9.0 & 9.8\\ 
 15 & 0.6 & $V$           & Nodes & 29667 & 26674 & 3832 & 4079 & 55357 & 43415 & 3346 & 3241\\ 
    &     &                & Time & 17.9 & 28.1 & 5.9 & 8.4 & 40.1 & 43.8 & 5.7 & 7.3\\ 
    &     & $\frac{V}{2}$ & Nodes & 55815 & 110729 & 6945 & 5019 & 67552 & 72351 & 5881 & 4842\\ 
    &     &                & Time & 44.9 & 153.3 & 16.1 & 15.5 & 62.6 & 118.9 & 15.4 & 15.4\\
\hline
    &     &               & \#i(\% relgap) & 5(0.00) & 3(9.25) & 5(0.00) & 5(0.00) & 5(0.00) & 5(0.00) & 5(0.00) & 5(0.00)\\ 
    &     & $\emptyset$   & Nodes & 988054 & 466550 & 39281 & 55712 & 629510 & 772783 & 34537 & 52009\\ 
    &     &                & Time & 2420.9 & 1604.6 & 289.9 & 508.0 & 1387.1 & 2337.9 & 212.9 & 440.0\\ 
    &     &               & \#i(\% relgap) & 4(2.50) & 5(0.00) & 5(0.00) & 5(0.00) & 5(0.00) & 4(2.50) & 5(0.00) & 5(0.00)\\ 
 20 & 0.6 & $V$           & Nodes & 1020377 & 1259144 & 32774 & 55886 & 1127201 & 992055 & 31273 & 29355\\ 
    &     &                & Time & 1354.4 & 2973.9 & 191.7 & 469.5 & 1704.1 & 2224.3 & 166.7 & 209.6\\ 
    &     &               & \#i(\% relgap) & 4(2.86) & 2(12.30) & 5(0.00) & 5(0.00) & 4(5.00) & 2(15.52) & 5(0.00) & 5(0.00)\\ 
    &     & $\frac{V}{2}$ & Nodes & 1168007 & 92724 & 40389 & 38547 & 1709270 & 84568 & 60645 & 105901\\ 
    &     &                & Time & 2998.5 & 253.6 & 234.1 & 313.5 & 3103.4 & 328.0 & 469.8 & 1271.8\\ 
\hline
    &     & $\emptyset$   & Nodes & 111586 & 38238 & 5197 & 4097 & 59593 & 46207 & 7068 & 5288\\ 
    &     &                & Time & 127.3 & 61.8 & 17.2 & 14.7 & 69.6 & 104.6 & 22.3 & 23.6\\ 
 20 & 0.8 & $V$           & Nodes & 50368 & 18847 & 2134 & 1799 & 40963 & 21796 & 2529 & 2012\\ 
    &     &                & Time & 54.3 & 31.2 & 5.5 & 5.1 & 38.2 & 42.0 & 7.3 & 8.3\\ 
    &     & $\frac{V}{2}$ & Nodes & 217452 & 47566 & 4677 & 3974 & 80027 & 59136 & 4443 & 4026\\ 
    &     &                & Time & 265.6 & 94.1 & 13.3 & 13.2 & 78.2 & 126.2 & 14.7 & 15.9\\ 
\hline
    &     & $\emptyset$   & Nodes & 516552 & 195865 & 12343 & 9223 & 224284 & 140744 & 8263 & 9559\\ 
    &     &                & Time & 958.0 & 789.2 & 60.4 & 94.5 & 345.1 & 603.8 & 50.0 & 82.3\\ 
 25 & 0.8 & $V$           & Nodes & 661557 & 173850 & 9609 & 8180 & 518909 & 103340 & 11362 & 12139\\ 
    &     &                & Time & 1388.1 & 676.6 & 53.9 & 55.8 & 1106.0 & 415.8 & 74.8 & 83.3\\ 
    &     & $\frac{V}{2}$ & Nodes & 210278 & 201639 & 12131 & 14444 & 540362 & 149792 & 8258 & 7914\\ 
    &     &                & Time & 365.2 & 1005.3 & 70.7 & 161.8 & 1000.2 & 729.8 & 65.0 & 66.7\\ 
\hline
    &     &               & \#i(\% relgap) & 1(39.17) & 3(22.50) & 5(0.00) & 4(19.13) & 1(33.95) & 2(34.17) & 5(0.00) & 5(0.00)\\ 
    &     & $\emptyset$   & Nodes & 967876 & 612289 & 73141 & 154230 & 380751 & 459696 & 63123 & 105989\\ 
    &     &                & Time & 6940.1 & 4232.1 & 1218.3 & 3075.6 & 2571.8 & 3565.8 & 1231.8 & 3017.8\\ 
    &     &               & \#i(\% relgap) & 5(0.00) & 4(10.00) & 5(0.00) & 5(0.00) & 4(3.33) & 3(7.33) & 5(0.00) & 5(0.00)\\ 
 30 & 0.8 & $V$           & Nodes & 1013007 & 620869 & 24843 & 17513 & 616678 & 510177 & 20820 & 19282\\ 
    &     &                & Time & 2934.5 & 2264.6 & 313.0 & 265.3 & 1723.4 & 3719.1 & 267.5 & 390.8\\ 
    &     &               & \#i(\% relgap) & 2(36.36) & 5(0.00) & 5(0.00) & 5(0.00) & 3(15.62) & 4(23.33) & 5(0.00) & 5(0.00)\\ 
    &     & $\frac{V}{2}$ & Nodes & 1811972 & 381841 & 29016 & 30301 & 1171870 & 550678 & 31445 & 44132\\ 
    &     &                & Time & 5847.9 & 2255.6 & 383.9 & 499.5 & 4343.0 & 3122.2 & 665.3 & 692.2\\ 
\hline
    &     &               & \#i(\% relgap) & 95(3.94) & 97(2.57) & \textbf{105}(0.00) & 104(0.91) & 96(2.84) & 95(3.95) & \textbf{105}(0.00) & \textbf{105}(0.00) \\ 
    & Av. &               & Nodes & 674570 & 426386 & 53057 & 48456 & 568656 & 360284 & 59731 & \textbf{41481} \\ 
    &     &                & Time & 1164.6 & 1033.3 & \textbf{204.8} & 308.3 & 950.6 & 932.0 & 243.9 & 361.7 \\ 
		&     &                & Time$^*$ & 1739.5 & 1503.1 & \textbf{204.8} & 373.9 & 1486.3 & 1529.0 & 243.9 & 361.7 \\
\hline
\end{tabular}
\caption{Comparison of formulations}  \label{tab:1}
\end{table}

\subsection{Separation and Cutting-plane algorithm}

In order to know if the addition of our cuts is effective during the optimization, we have performed a test by running CPLEX over
the instances of 7 combinations: $(15,0.2)$, $(20,0.4)$, $(20,0.6)$, $(20,0.8)$, $(25,0.6)$, $(25,0.8)$, $(30,0.8)$.
We have used formulation $F_3$ and considered these algorithms:
\begin{itemize}
	\item B\&B - no cuts are generated,
	\item B\&C$_1$ - only inequalities of Type I are generated,
	\item B\&C$_2$ - inequalities of the two types are generated.
\end{itemize}
The separation of Type I inequalities is performed 10 times in the root node, twice in nodes with depths 1 and 2, and once in nodes
with depths 3 to 10. In B\&C$_2$, the routine that separates Type II ineq.~is executed after the one for Type I in nodes with depth at most 5.

Below, we describe the implementation of both routines. They are designed not to strictly follow the hypotheses of facet-defining results, but
to run fast. The current fractional solution is denoted by $(x^*, y^*)$.

\subsubsection{Separation of Type I inequalities}

Before starting the optimization, create sets
$$\mathcal{W}_u = \{ w \in N\la u\ra : |N\la w\ra| \geq 2 ~\textrm{and, for all}~ v \in N\la u\ra \setminus \{w\}, w \antes v ~\textrm{and}~ v \antes w\}$$
for each $u \in V$ (these sets are related to hypotheses H3-H4 of Prop.~\ref{GENERALINEQ2}).
Each time the separation routine is invoked, assign $\mathcal{A} \leftarrow V$.
Then, for every $u \in V$ and $w \in \mathcal{W}_u \cap \mathcal{A}$ do the following. Set $sum \leftarrow y^*_{w1}$.
For all $i = 2,\ldots,m$, do $sum \leftarrow sum + y_{wi}$ and check whether $x^*_{ui} + sum > 1.1$. In that case, add
$x_{ui} + \sum_{j=1}^i y_{w j} \leq 1$ as a cut and remove $w$ from $\mathcal{A}$.

The set $\mathcal{A}$ stores those vertices  ``$w$'' not used by cuts from previous iterations, thus preventing the generation of
cuts with similar support.

\subsubsection{Separation of Type II inequalities}

Before starting the optimization, create sets
\begin{multline*}
  \mathcal{W}_{u_1 u_2} = \{ w \in \mathcal{W}_{u_1} \cap \mathcal{W}_{u_2} :
     ~\textrm{there are}~ z_1 \in N\la u_1\ra \setminus N\la u_2\ra, z_2 \in N\la u_2\ra \setminus N\la u_1\ra\\
		\textrm{such that}~ N\la w\ra \setminus (\{u_2\} \cup N\la z_1\ra) \neq \emptyset,
		N\la w\ra \setminus (\{u_1\} \cup N\la z_2\ra) \neq \emptyset\}
\end{multline*}
for each pair $\{u_1, u_2\} \subset V$, which are related to hypotheses of Prop.~\ref{FAM2}.
The separation routine is executed immediately after the separation of Type I ineq.~and makes use of the vertices that remains in
$\mathcal{A}$. 
For every $\{u_1, u_2\} \subset V$ and $w \in \mathcal{W}_{u_1 u_2} \cap \mathcal{A}$ do the following. Set $sum \leftarrow y^*_{w1}$.
For all $i = 2,\ldots,m$, do $sum \leftarrow sum + y_{wi}$ and check whether $x^*_{u_1 i} \notin \mathbb{Z}$ and
$x^*_{u_2 i} \notin \mathbb{Z}$. In that case, for all $k = 1,\ldots,i$, if $y^*_{wk} \notin \mathbb{Z}$ then check whether
$x^*_{u_1 i} + x^*_{u_2 i} + sum + \sum_{v \in N^{\cup}} y_{vk} > 2.2$ and, in that case, add
$x_{u_1 i} + x_{u_2 i} + \sum_{j=1}^i y_{wj} + \sum_{v \in N^{\cup}} y_{vk} \leq 2$ as a cut and remove $w$ from $\mathcal{A}$.

\subsubsection{Comparing algorithms}

Table \ref{tab:2} reports the results. Each row displays the parameters obtained by the three algorithms for a given combination $(n,p)$.
As in Table \ref{tab:1}, ``\#i(\% relgap)'' refers to the number of solved instances and the average of the percentage
relative gap, ``Nodes'' and ``Time'' are the averages of the number of evaluated nodes and the elapsed time (in seconds), respectively, over
the solved instances. The last two rows summarize the results for all combinations, with best values in boldface.
Here, the last row (``Av$^*$'') reports the average of the elapsed time over all instances, by taking 7200 sec. for those ones that are
not solved, see the definition of Time$^*$ above.

\begin{table}[t] \centering \footnotesize
\begin{tabular}{|@{\hspace{3pt}}c@{\hspace{3pt}}|@{\hspace{3pt}}c@{\hspace{3pt}}|@{\hspace{3pt}}c@{\hspace{3pt}}|@{\hspace{3pt}}c@{\hspace{3pt}}|@{\hspace{3pt}}c@{\hspace{3pt}}|@{\hspace{3pt}}c@{\hspace{3pt}}|@{\hspace{3pt}}c@{\hspace{3pt}}|@{\hspace{3pt}}c@{\hspace{3pt}}|@{\hspace{3pt}}c@{\hspace{3pt}}|@{\hspace{3pt}}c@{\hspace{3pt}}|@{\hspace{3pt}}c@{\hspace{3pt}}|}
\hline
   & & \multicolumn{3}{c|}{\#i(\% relgap)} & \multicolumn{3}{c|}{Nodes} & \multicolumn{3}{c|}{Time} \\
 $n$ & $p$ & B\&B & B\&C$_1$ & B\&C$_2$ & B\&B & B\&C$_1$ & B\&C$_2$ & B\&B & B\&C$_1$ & B\&C$_2$ \\
\hline
15 & 20 & 15(0.00) & 15(0.00) &  15(0.00) & 219518 & 110810 &  95554 &  360.2 &  283.0 &  238.6 \\
20 & 40 & 11(8.04) & 11(4.88) &  13(1.67) & 542344 & 184416 & 130556 & 3543.3 & 2513.1 & 2302.9 \\
20 & 60 & 15(0.00) & 15(0.00) &  15(0.00) &  37482 &  21011 &  28611 &  238.6 &  199.9 &  386.2 \\
20 & 80 & 15(0.00) & 15(0.00) &  15(0.00) &   4003 &   3257 &   3891 &   12.0 &   11.2 &   12.6 \\
25 & 60 & 6(29.77) & 11(8.72) & 13(10.59) & 244484 &  93348 & 121882 & 2498.7 & 2380.6 & 3053.3 \\
25 & 80 & 15(0.00) & 15(0.00) &  15(0.00) &  11361 &   9318 &  10779 &   61.7 &   56.7 &   71.3 \\
30 & 80 & 15(0.00) & 15(0.00) &  15(0.00) &  42333 &  34746 &  37530 &  638.4 &  713.5 &  825.7 \\
\hline
 & Av. & 92(5.40) & 97(1.94) & \textbf{101}(1.75) & 132099 & 59201 & \textbf{58685} & 800.4 & \textbf{750.5} & 917.3 \\
 & Av.$^*$ &   &   &   &   &   &   & 1593 & 1242 & \textbf{1157} \\
\hline
\end{tabular}
\caption{Comparison of algorithms} \label{tab:2}
\end{table}

Note that both B\&C$_1$ and B\&C$_2$ outperform B\&B, with B\&C$_1$ presenting the lowest average of CPU time per solved instance and
B\&C$_2$ being the best one in terms of number of solved instances, relative gap, number of nodes and the parameter reported in the last row.
We conclude that the addition of the new inequalities greatly reinforces the linear relaxations, thus helping to
decrease the number of nodes in the B\&B tree and the consumed time.

\section{Conclusions}

We have introduced the General Grundy Domination Problem (GGDP), which generalizes both the Grundy Domination and Grundy Total Domination problems.
We tackle such combinatorial optimizations problem by integer programming techniques.
To the best of our knowledge, this is the first mathematical programming approach for Grundy domination numbers.
The basis of our work are novel ILP formulations relying on the concept of legal sequence and the study of strong valid inequalities.
Besides the constraints of the formulations and its generalizations, we were able to present a large class of valid inequalities and necessary/sufficient conditions for facet-definition. We observed that the clutter property plays an important role in this context.

Besides the theoretical achievements, we also reported computational experiments that aimed at comparing the performances of the formulations in a pure B\&B framework as well as the strength of the derived cuts. Regarding the formulations, we could attest the power of the symmetry-breaking constraints \eqref{RESTR8}-\eqref{RESTR9} in reducing both the number of nodes and the computational times. In addition, we could empirically demonstrate the strength of the cuts, even in restricted versions.


\end{document}